\documentclass[a4paper]{scrartcl}
\usepackage[latin1]{inputenc}
\usepackage{amsmath}
\usepackage{amsfonts}
\usepackage{amssymb}
\usepackage{amsthm}
\usepackage{graphicx}
\usepackage{tabularx}
\usepackage{booktabs,dcolumn}
\usepackage[pdftex]{hyperref}
\usepackage{natbib}
\newcolumntype{C}{>{\centering\arraybackslash}X}

\newcommand{\bysame}{\hskip.3em \leavevmode\rule[.5ex]{3em}{.3pt}\hskip0.5em}

\DeclareMathOperator*{\argmin}{arg\,min}
\DeclareMathOperator*{\E}{E}
\DeclareMathOperator*{\Var}{Var}
\DeclareMathOperator*{\Std}{Std}
\DeclareMathOperator*{\Cov}{Cov}

\newtheorem{theorem}{Theorem}
\newtheorem{proposition}[theorem]{Proposition}
\newtheorem{lemma}[theorem]{Lemma}
\theoremstyle{definition}
\newtheorem{definition}[theorem]{Definition}

\title{{\normalsize The performance and efficiency of}\\ \Huge Threshold Blocking} 
\author{Fredrik S\"avje\thanks{Department of Economics, Uppsala University. \mbox{\texttt{\url{http://fredriksavje.com}}}}}

\begin{document}
\maketitle

\begin{abstract}
	A common method to reduce the uncertainty of causal inferences from experiments is to assign treatments in fixed proportions within groups of similar units: blocking. Previous results indicate that one can expect substantial reductions in variance if these groups are formed so to contain exactly as many units as treatment conditions. This approach can be contrasted to threshold blocking which, instead of specifying a fixed size, requires that the groups contain a minimum number of units. In this paper, I investigate the advantages of respective method. In particular, I show that threshold blocking is superior to fixed-sized blocking in the sense that it always finds a weakly better grouping for any objective and sample. However, this does not necessarily hold when the objective function of the blocking problem is unknown, and a fixed-sized design can perform better in that case. I specifically examine the factors that govern how the methods perform in the common situation where the objective is to reduce the estimator's variance, but where groups are constructed based on covariates. This reveals that the relative performance of threshold blocking improves when the covariates become more predictive of the outcome.
\end{abstract}

\section{Introduction}

Randomly assigning treatments to units in an experiment guarantees that one is expected to capture treatment effects without error. Randomness is, however, a treacherous companion. It lacks biases but is erratic. Once in a while it produces assignments that by any standard must be considered absurd---giving treatment only to the sickest patients or reading aids only to the best students. While we can be confident that all imbalances are accidental, once they are observed, the validity of one's findings must still be called into question. Any reasonably designed experiment should try to avoid this erratic behavior and doing so inevitably reduces randomness. 

This paper contributes to a longstanding discussion on how such behavior can be prevented. This discussion originates from a debate whether even the slightest imbalances should be accepted to facilitate randomization \citep{Student1938}; if imbalances are problematic, it is only natural to ask why one would not do everything to prevent them.\footnote{While Gosset argued that a balanced experiment was to be preferred over one that was only randomized, his ideal seems to be to combine both. See, e.g., the third footnote of \citet{Student1938}.} The realization that no other method can provide the same guarantee of validity has, however, lead to an overwhelming agreement that randomization is the key to a well-designed experiment and shifted the focus to how one best tames it. As seemingly innocent changes to the basic design can break the guarantee, or severely complicate the analysis, any modification has to be done with care. Going back to at least \citet{Fisher1926}, \emph{blocking} has been the default method to avoid the absurdities that randomization could bring while retaining its desirable properties.

In its most stylized description, blocking is when the scientist divides the experimental sample into groups, or \emph{blocks}, and assigns treatment in fixed proportions within blocks but independently between them. If one is worried that randomization might assign treatment only to the sickest patients, one should form these groups based on health status. By doing so, one ensures that each group will be split evenly between the treatment conditions and thereby avoids that only one type of patients are treated---the treatment groups will, by construction, be balanced with respect to health status.

The currently considered state of the art blocking method is paired matching \citep{Greevy2004}, or \emph{paired blocking}, where one forms the blocks so that they contain equally many units as treatment conditions. Paired blocking is part of a class of methods that re-interprets blocking as an optimization problem. Common to these methods is that one specifies some function to describe the desirability of the blockings and forms the blocks so to reach the best possible blocking according to the measure. The scientist typically seeks covariate balance between the treatment groups, in which case the objective function could be some aggregate of a distance metric within the blocks.

In this paper, I will discuss a development of the paired blocking method introduced in \citet{Higgins2015algorithm}: \emph{threshold blocking}. This method should be contrasted to any \emph{fixed-sized blocking}, of which paired blocking is a special case. The two methods differ in the structure they impose on the blocks, in particular the size constraints. One often want to ensure that at least a certain number---nearly always some multiple of the number treatment conditions---of units are in each block as less can lead to analytical difficulties. Fixed-sized blocking ensures that this is met by requiring that all blocks are of a certain size. Threshold blocking recognizes that larger blocks than the requirement are less problematic than smaller, and that they even can be beneficial. Instead of forcing each block to be of the same size, it only requires a minimum number of units.

My first contribution is to show that relaxing the size restriction will lead to weakly better blockings: for any objective function and sample, the optimal threshold blocking can be no worse than the optimal fixed-sized blocking. This result follows directly from the fact that the search set of threshold blocking is a superset of the fixed-sized search set. While smaller blocks is preferable for most common objectives, and thus seemingly rendering the added flexibility of threshold blocking immaterial, allowing for a few locally suboptimal blocks can prevent very awkward compositions in other parts of the sample.

The interpretation of the blocking problem as a classical optimization problem is not fitting for all experiments. We are, for example, in many instances interested in the variance of our estimators and employ blocking to reduce it. The variance of different blockings can, however, not be calculated or even estimated beforehand. The objective function of true interest is unknown. We must instead use some other function, a \emph{surrogate}, to form the blocks. The performance of threshold blocking depends on whether a surrogate is used, and then on its quality.

With a known objective function we can always weigh the benefits and costs, so that threshold blocking finds the best possible blocking. If we instead use a surrogate, it might not perfectly correspond to the behavior of the unknown objective. A perceived benefit might not reflect the actual result. While we still can expect threshold blocking to be beneficial in many settings, it is possible that the surrogate is misleading in a way that favors fixed-sized blockings. In general, when the surrogate is of high quality (i.e., unlikely to be misleading), the added flexibility of threshold blocking will be beneficial.

The factors that govern the surrogate's quality are specific to each pair of objective and surrogate. The main contribution of this study is to investigate which factors are important in the common case where covariate balance is used as a surrogate for variance reduction. In particular, I show that the variance resulting from any blocking method can be decomposed into several parts of which two are affected by the properties of the method. The first part is that the variance will be lower when there is more balance in the expected potential outcomes, which is a function of the covariates. As covariate balance is observable at the time of blocking, threshold blocking will lead to the greatest improvement with respect to this part.

The second part accounts for variation in the number of treated in the blocks, which increases estimator variance. Fixed-sized blocking will construct blocks that are multiples of the number of treatment condition. This ensures that equally many units will be assigned to each treatment group making this factor immaterial. The flexibility of threshold blocking can, however, introduce such variation and subsequently lead to increased variance. The relative performance between the methods thus depends on whether the covariates are predictive enough, and whether the relevant type of covariate balance is considered, so that improvement in the surrogate offsets the increase expected from variability in the number of treated.

These results could put the current view of fixed-sized blocking as the default blocking method into question. The use of fixed-sized blocking, and paired blocking in particular, is often motivated by that it never can result in a higher unconditional variance than when no blocking is used, and that it leads to the lowest possible unconditional variance of all blocking methods. This study show that we can expect threshold blocking to outperform paired blocking in many situations, and, in particular, in situations where blocking is likely to be beneficial (i.e., when covariates are predictive of the outcome). With a threshold design there is, however, no longer a guarantee that the variance is no larger than with no blocking.

In the next section, I will introduce threshold blocking in more detail and discuss how it relates to other methods to achieve balance in experiments. In Section \ref{thresholdb-adv}, I formally describe the two blocking methods and prove that threshold blocking outperforms fixed-sized blocking when the objective function is known and discuss the consequences of an unknown function. Section \ref{sec-variance} looks specifically at the case when reduction in unconditional variance is of interest. This is followed by a small simulation study in a slightly more realistic setting, and Section \ref{sec-conclude} concludes.

\section{Threshold blocking}

A useful way to understand blocking and other methods that aim to make the treatment groups more alike is to consider how they introduce dependencies in the assignment of treatments. We can then loosely order the methods along a continuum based on the degree of introduced dependence. Specifically, to make treatment groups more similar we want to impose a negative correlation in the treatment assignments among similar units; if a unit is in one treatment group, units that are similar to the first are likely to be in other groups. At one extreme of the continuum, treatment is assigned using a coin flip for each unit and, subsequently, each unit's treatment is independent all the others'. At the other extreme, all treatments are perfectly correlated so that all assignments are determined by a \emph{single} coin flip regardless of the sample size.\footnote{Re-randomization methods \citep{Morgan2012} are hard to place on this continuum. Depending on the re-randomization criterion, they could introduce any level of dependence: if the criterion is non-binding, there would be no dependence, and with a (symmetric) criterion so strict so that only two assignments are accepted, we are at the other extreme. What is common to all re-randomization methods is that, if any dependence is introduced, it is generally so complex that it is analytically inaccessible, and one must rely on permutation-based inferences.}

Two factors change along this continuum. The more appropriate dependence that is introduced, the more accurate the inference will be as we can impose the desirable correlation structure.\footnote{Obviously, any dependence is not useful---to impose correct type requires a lot of information.} The correlation structure we impose must, however, be accounted for in the analysis. While this is generally trivial for point estimation, estimation of standard errors or null distributions can be hard, if not impossible, without very restrictive assumptions. In general, the more dependence we introduce the harder this problem becomes. Our position on the continuum is largely a trade-off between achieving balance between treatment groups (and thereby accuracy), and analytical ease.

There seems to be consensus that neither extreme is a good choice for most experiments. Independent treatment assignment makes for almost trivial analysis, but with only a small added complexity, accuracy can be improved often considerably (e.g., by using complete randomization). At the other extreme, perfectly correlated assignment will, when possible, minimize variance \citep{Kasy2013}, but it makes essentially all reasonable uncertainty measures unattainable. The sample space of the estimator, conditional on the drawn sample, is here two points of which we observe one---not even a permutation based approach would be possible. Instead, both theory and applications have focused on the middle of the continuum where the major blocking methods are positioned.

Blocking methods recognize that a negative correlation is most useful between certain units. Instead of imposing a large complicated correlation structure in the whole sample, they remove the less needed dependencies to keep only the important ones. By assigning treatment in fixed proportions within the blocks, a strong negative correlation is imposed within groups of units that are very similar, but it keeps assignment independent across groups. With independent blocks the analysis is considerably simpler than if all assignments were correlated. However, as expected from the trade-off, we could improve accuracy further---for example by introducing dependencies also between blocks, so that an unbalanced assignment in one block tended to be counteracted by a slight imbalance in the opposite direction in another---but this would also obfuscate the analysis.

What differentiates blocking methods is how they form blocks. The original blocking methods partitioned the sample into perfectly homogeneous groups based on some categorical variable, usually an enumeration of strata of discrete covariates.\footnote{Some authors still use ``blocking'' to exclusively refer to this type of method. In this paper, I will take all methods that assign treatment with high dependence within pre-constructed groups of units, but independently across them, to be blocking methods.} In cases with very low dimensional data, this method works well; one simply form blocks just as the sample is naturally clustered. However, with more information, all observations will typically be unique. As there then exist no homogeneous groups, this approach is no longer possible. Inspired by the multivariate matching problem in observational studies \citep{Cochran1973,Rosenbaum1989}, modern blocking methods construct blocks based on a distance metric or some other function indicating the degree of similarity between units \citep{Greevy2004}. Doing so, the problem is transformed into an optimization problem. When homogeneous groups do not exist in the sample, these methods set out to find the blocking that, while not being perfect, is the best possible. The blocking methods considered in this paper are all in the class.

Much of the recent work has focused on, what I will refer to as, \emph{fixed-sized blocking}, which is part of this class of methods. With this method, blocks are constructed so to minimize the objective function subject to that they all are of a certain size. There are several reasons why one would want to impose a size restriction. Primarily, many estimators are based on the within-block differences of the average outcome in the treatment conditions. If the blocks do not contain at least as many units as treatment conditions, these estimators will be undefined. Too large blocks are, however, not desirable either. Returning to our continuum, if we want to maximize the expected similarity of the treatment groups, we want to keep the blocks as small as possible as this maximizes the negative correlation between units assigned to the same block. These two objectives together---keeping block sizes as small as possible while ensuring a certain number units in the blocks---suggests a fixed-sized blocking.

Threshold blocking, as introduced in \citet{Higgins2015algorithm}, is another subclass of this class of methods. It differs from fixed-sized blocking only in that it imposes a minimum block size rather than a strict size requirement. In many ways, the difference is parallel to the difference between full matching and one-to-one (or one-to-$k$) matching in observational studies \citep{Rosenbaum1991,Hanson2004}: threshold blocking allows for a more flexible structure of the blocks and can therefore find a more desirable solution. Or, with the interpretation as an optimization problem, threshold blocking extends the problem's search space.

On our continuum, the two subclasses are positioned approximately at the same place, in the sense that the amount of dependence does not differ much. The difference lies in the type of correlation they introduce. Keeping the blocks exactly at the specified size, fixed-sized blocking ensures that units within the same block are correlated to the greatest possible extent. However, in much the same way that one-to-one matching often is forced to do bad matches and forgo good matches due to the required match structure, the strict size requirement often forces fixed-sized blocking to make two units' assignment independent even if they ideally should be highly correlated. Conversely, it sometimes must impose a high correlation when they ideally should be independent. Threshold blocking allows for slightly less correlation between some units (i.e., bigger blocks) if this avoids such situations. It still recognizes that a minimum size is very beneficial due to the restrictive analytical problems that otherwise would follow and achieves the flexibility by allowing for bigger, but not smaller, blocks.

To illustrate this difference, consider when the specified block size is two and there are three very similar units in the sample. Fixed-sized blocking would here be forced to pick two of the units to have perfectly correlated assignments but which are uncorrelated with the third unit. The third unit, in turn, would be forced to be perfectly correlated with some other, less similar, unit. Threshold blocking has the option to put all three units in the same block. There will be less correlation between the two previously blocked units, but they will no longer be independent with respect to the third unit.

This study is part of a growing literature on how to best ensure balance in experiments. Iterating the introduction, methodologists seem to agree that one should try to balance experiments whenever possible \citep{ImaiKing2008,Rubin2008}, but there is still an active discussion on how one best does so. Apart from a large strand of the literature focusing on the algorithmic aspects of the problem \citep[see, e.g., ][]{Lu2011,Moore2012,Higgins2015algorithm}, some recent contributions have discussed more general properties of different blocking methods as I do in this paper.

Closely related is an investigation by \citet{Imbens2011} that partly focuses on the optimal block size. Specifically, he questions whether paired blocking is the ideal blocking method. To maximize assignment dependence, and thus accuracy, we want to keep the blocks as small as possible just like paired blocking does. Imbens notes that, while this would lead to the lowest variance, estimation of conditional standard errors is quite intricate when the blocks only assign a single unit to each treatment. For this reason, he recommends that blocks contain at least twice the number units as treatment conditions. While also being concerned with the block size, Imbens investigates the optimal block size requirement rather than how it best should be imposed. In the analogy with matching, Imbens' study is closer to which the optimal $k$ is in one-to-$k$ matching, instead of its performance relative to full matching as in this study.

Also related to my inquiry, while not examining the block size, are a few recent papers on the optimal balancing strategy. \citet{Kasy2013} discusses a situation where one is blessed with precise priors of the relation between the covariates and the (unobserved) outcomes. He shows that when such information is available the optimal design, with respect to mean square error of the treatment effect estimator, is to minimize randomization (i.e., not to randomize at all or only do so with a single coin flip). In other words, he advocates for the previously discussed extreme position on our continuum. While we indeed can expect this to minimize the uncertainty of the point estimates, the analytical challenges that inevitably follow will oftentimes be too troublesome. For example, most conditional standard errors are impossible to estimate and unconditional variances require strong assumptions.

Related to \citet{Kasy2013} is a study by \citet{Kallus2013}. He shows that, using a minimax criterion and interpreting experiments as a game against a malevolent nature, all blocking-like methods will produce a higher variance than complete randomization unless we have some information on the relation between the covariates and the outcome. He goes on showing how to derive the optimal design given an information set, and he shows that certain information sets lead to the classical designs. While Kallus' set-up is not directly applicable to threshold blocking (as his condition 2.3 prescribes fixed-sized blocks), there is no reason to expect that the results would not carry over also to the current setting. One can, however, discuss whether his problem formulation is relevant for the typical experiment. The result hinges on the use of minimax criteria. It is not clear why we would only be interested in the performance under the worst imaginable sample draw. Changing the criteria to something less risk-averse, e.g., the average performance, the results no longer hold. Nevertheless, \citet{Kallus2013} clearly illustrates the important role that the outcome model, and our information of it, plays in blocking problems.

Last, \citet{Barrios2014} investigate the optimal (surrogate) objective function to use with paired blocking when interest is in variance reduction. He demonstrates that if we have access to the conditional expectation function (CEF) of the outcome, and under a weak version of the constant treatment effect assumption, it is best to seek balance in the predicted outcomes from the CEF. Estimator variance is related to how much to treatment groups differ with respect to the potential outcomes. In order to lower variance, we thus want to impose a negative correlation between units with similar potential outcomes. We cannot observe these outcomes beforehand, but the best predictor of them (i.e., the CEF) will form an excellent surrogate. While Barrios' study is restricted to paired blocking, as hinted by the investigation in Section \ref{sec-decomp}, his results likely extend also to other fixed-sized blockings and to threshold blocking. Of course, using this surrogate requires us to have access to detailed information about the outcome model.

\section{The advantage of threshold blocking} \label{thresholdb-adv}

Let $\mathbf{U}=\{1,2,\cdots,n\}$ be a set of unit indices representing sample of $n$ units in an experiment.

\begin{definition}
	A \emph{block} is a non-empty set of unit indices. A \emph{blocking} of $\mathbf{U}$ is a set of blocks, $\mathbf{B}=\{\mathbf{b}_1, \mathbf{b}_2, \cdots, \mathbf{b}_m\}$, such that:
	\begin{enumerate}
		\item $\forall\; \mathbf{b} \in \mathbf{B},\mathbf{b} \neq \emptyset$,
		\item $\bigcup_{\mathbf{b}\in\mathbf{B}}\mathbf{b} = \mathbf{U}$,
		\item $\forall\; \mathbf{b}_i, \mathbf{b}_j \in \mathbf{B}, \mathbf{b}_i \neq \mathbf{b}_j \Rightarrow \mathbf{b}_i \cap \mathbf{b}_j = \emptyset$,
	\end{enumerate}
	In other words, a blocking is a collection of blocks so that all units are in exactly one block. 
\end{definition}
\begin{definition} \label{sharpdef}
	A \emph{fixed-sized blocking} of size $S$ of $\mathbf{U}$ is a blocking where all blocks contain exactly $S$ units: $\forall\;\mathbf{b} \in \mathbf{B}, |\mathbf{b}| = S$.
\end{definition}
\begin{definition} \label{defthreshold}
	A \emph{threshold blocking} of size $S$ of $\mathbf{U}$ is a blocking where all blocks contain at least $S$ units: $\forall\;\mathbf{b} \in \mathbf{B}, |\mathbf{b}| \geq S$.
\end{definition}
Let $\mathbf{A}$ denote the set of all possible blockings of $\mathbf{U}$. Let $\mathbf{A}_F$ and $\mathbf{A}_T$ denote the sets of all admissible fixed-sized and threshold blockings of a certain size, $S$:
\begin{eqnarray*}
	\mathbf{A}_F &=& \{\mathbf{B}\in \mathbf{A}: \forall\, \mathbf{b} \in \mathbf{B}, |\mathbf{b}| = S\},
	\\
	\mathbf{A}_T &=& \{\mathbf{B}\in \mathbf{A}: \forall\, \mathbf{b} \in \mathbf{B}, |\mathbf{b}| \geq S\}.
\end{eqnarray*}
Note that for all samples where $|\mathbf{U}|$ is not a multiple of the size requirement, $\mathbf{A}_F$ will be the empty set as no blocking fulfills Definition \ref{sharpdef}. One trivial advantage of threshold blocking is that it can accommodate any sample size. Since no performance comparison can be done in that situation, I will restrict my attention to situations where $\mathbf{A}_F$ is not empty.

Consider some objective function that maps from blockings to the real numbers, $f : \mathbf{A} \rightarrow \mathbb{R}$, where a lower value denotes a more desirable blocking.\footnote{We are currently not concerned about exactly what this function is---it suffices to note that the same objective function can be used for both fixed-sized and threshold blocking.}

\begin{definition} 
	An optimal blocking, $\mathbf{B}^*$, in a set of admissible blockings, $\mathbf{A}'$, fulfills:
	\begin{equation*}
		f(\mathbf{B}^*) = \min_{\mathbf{B}\in\mathbf{A}'}f(\mathbf{B}).
	\end{equation*}
\end{definition}

Whenever the sample is finite, the number of possible blocking ($\mathbf{A}$) is also finite which bounds the number of admissible blockings in any blocking problem. This ensures that a solution exists for all blocking problems as long as at least one valid blocking exists. Optimal blockings need, however, not be unique.

Let $\mathbf{B}_F^*$ and $\mathbf{B}_T^*$ denote optimal fixed-size and threshold blockings:
\begin{eqnarray*}
	\mathbf{B}_F^* &=& \argmin_{\mathbf{B}\in\mathbf{A}_F}f(\mathbf{B}),
	\\
	\mathbf{B}_T^* &=& \argmin_{\mathbf{B}\in\mathbf{A}_T}f(\mathbf{B}).
\end{eqnarray*}

\begin{lemma} \label{subset}
	For all samples and all $S$, the set of admissible fixed-sized blockings is a subset of the set of admissible threshold blockings:
	\begin{equation*}
		\mathbf{A}_F \subseteq \mathbf{A}_T.
	\end{equation*}
\end{lemma}
\begin{proof}
	All blockings in $\mathbf{A}_F$ contain blocks so that $|\mathbf{b}|=S$. These blockings also satisfy $\forall\, \mathbf{b} \in \mathbf{B}, |\mathbf{b}| \geq S$ which, by Definition \ref{defthreshold}, make them elements of $\mathbf{A}_T$.
\end{proof}

\begin{theorem} \label{noworse}
	For all samples, all objective functions and all $S$, the optimal threshold blocking can be no worse than the optimal fixed-sized blocking:
	\begin{equation*}
		f(\mathbf{B}_T^*)\leq f(\mathbf{B}_F^*).
	\end{equation*}
\end{theorem}
\begin{proof}
	This follows almost trivially from Lemma \ref{subset}. Assume $f(\mathbf{B}_T^*)> f(\mathbf{B}_F^*)$. This implies that $\mathbf{B}_F^* \not\in \mathbf{A}_T$ as otherwise $f(\mathbf{B}_T^*)$ would not be the minimum in $\mathbf{A}_T$. By Lemma \ref{subset} we have $\mathbf{B}_F^* \in \mathbf{A}_T$ and thus a contradiction. 
\end{proof}

\begin{theorem} \label{thres:thm:thres-better}
	There exist samples and objective functions for which threshold blocking is strictly better than fixed-sized blocking:
	\begin{equation*}
		f(\mathbf{B}_T^*)< f(\mathbf{B}_F^*).
	\end{equation*}
\end{theorem}
\begin{proof}
	I will prove the theorem with two examples. These will also act as an introduction to the subsequent discussion. While trivial objective functions suffice to prove the theorem, these examples are chosen to be similar to actual experiments albeit being greatly simplified.
	
	The first example is when we construct the blocking so to minimize the covariate distances between the units within the blocks. This is a common objective used in many actual experiments. The second example is when the objective function is the variance of the treatment effect estimator conditional on the observed covariates. While the unconditional variance is often considered when comparing blocking methods (as I do in later sections), the conditional version best mirrors the position of the scientist as the blocking is decided after covariates, but before outcomes, are observed. As blocking often is used to reduce uncertainty, the second example is closer to the purpose of blocking.
	
	In both examples, the sample consists of six units in an experiment with two treatment conditions. For both fixed-sized and threshold blocking, the block size requirement is two ($S=2$). There is a single binary covariate, $x_i$, which is observed before the blocks are constructed. In the drawn sample, half of units have $x_i=1$ and the other half have $x_i=0$. The only information on the units is the covariate values. All units that share the same value are therefore interchangeable, and blockings can be denoted simply by how it partitions covariates rather than units. For example, $\mathbf{B}=\{\{1,0\},\{1,0\},\{1,0\}\}$ denotes all blockings where each block contain two units with different covariate values. 
	
	For tractability, I will, when applicable, make three simplifying assumptions. First, I assume that the sample is randomly drawn from an infinite population. Second, that the treatment effect is constant over all units in the population. This implies that $y_i(1)= \delta + y_i(0)$ for some treatment effect, $\delta$, where $y_i(1)$ and $y_i(0)$ denote the two potential outcomes in the Neyman-Rubin Causal Model \citep{Neyman1923,Rubin1974}. Third, that the conditional variance of the potential outcomes is constant:
	\begin{equation*}
		\forall x, \sigma^2 = \Var[y_i(1)|x_i=x]=\Var[y_i(0)|x_i=x].
	\end{equation*}
	While these assumptions are unrealistic in most applications, they should not cloud the overarching intuitions that can be gained from the examples.
	
	\subsection{Example 1: Distance metric} \label{ex-dist-met}
	
	In this example the objective function is an aggregate of within block distances between the units based on the covariate. Euclidean and Mahalanobis distances are commonly used as metrics in blocking problems. With a single covariate, as here, the Mahalanobis distance is proportional to the Euclidean and thus produces the same blockings. For simplicity, I will opt for the Euclidean metric, and the distance between units $i$ and $j$ is given by $\sqrt{(x_i-x_j)^2}$. To aggregate the distances and get the objective function, $f(\mathbf{B})$, I will use the average within-block distance weighted by the block size:
	\begin{eqnarray*}
		f(\mathbf{B}) &=& \sum_{\mathbf{b}\in\mathbf{B}}\frac{n_\mathbf{b}}{n} \bar{d}_\mathbf{b},
		\\
		\bar{d}_\mathbf{b} &=& \sum_{i\in\mathbf{b}}\sum_{j\in\mathbf{b}}\frac{\sqrt{(x_i-x_j)^2}}{n_\mathbf{b}^2},
	\end{eqnarray*}
	where $n_\mathbf{b}\equiv |\mathbf{b}|$ is the number of units in block $\mathbf{b}$, and $\bar{d}_\mathbf{b}$ is the average Euclidean distance within that block.\footnote{This aggregation differs slightly from the one that is most commonly used with fixed-sized blocking: the sum of distances. Using the sum works well when the blocks have constant sizes across all considered blockings, in fact the two coincide in that case. When sizes differ between blockings, the sum can be misleading as the number of distances within a block grows exponentially in the block size. Nonetheless, there are examples where threshold blocking is strictly better than fixed-sized blocking also using the sum of distances as objective.}
	
	\begin{table}[htbp]
		\centering
		\caption{Values of the objective functions for different blockings.} \label{tab-bestblocking}
		\begin{tabularx}{0.8\textwidth}{@{\vrule height11pt  width0pt} lXCC}
			&  & \multicolumn{2}{c}{Objectives}  \\ \cline{3-4}
			\\ [-2.3ex]
			Blocking ($\mathbf{B}$) & Valid for & Distance & Variance \\ \midrule
			\\ [-2.5ex]
			$\{\{1,0\}, \{1,0\}, \{1,0\}\}$ & Both & 0.500 & 1.333 \\
			$\{\{1,1\}, \{1,0\}, \{0,0\}\}$ & Both & 0.167 & 0.889 \\
			$\{\{1,1,1\}, \{0,0,0\}\}$ & Threshold & 0 & 0.750 \\
			$\{\{1,1,0\}, \{1,0,0\}\}$ & Threshold & 0.444 & 1.250 \\
			$\{\{1,1,1,0\}, \{0,0\}\}$ & Threshold & 0.250 & 0.889 \\
			$\{\{1,1,0,0\}, \{1,0\}\}$ & Threshold & 0.500 & 1.185 \\
			$\{\{1,0,0,0\}, \{1,1\}\}$ & Threshold & 0.250 & 0.889 \\
			$\{\{1,1,1,0,0,0\}\}$ & Threshold & 0.500 & 1.067 \\ [0.5ex]
			\bottomrule
			\\ [-1.9ex]
			\multicolumn{4}{p{0.76\textwidth}}{\footnotesize \textit{Note:} The table presents values of the objective functions resulting from different blockings. Each row represents a blocking, where only the first two rows are valid fixed-sized blockings. The third column presents the values when the aggregated distance metric is used as objective, as discussed in Section \ref{ex-dist-met}. The fourth column presents the values when the conditional variance ($\Var(\hat{\delta}|\mathbf{x},\mathbf{B})$) is used, as described in Section \ref{ex-condvar}.} \\
		\end{tabularx}
	\end{table}

	Using this function we can calculate the average distance of each blocking and thereby rank them. There are two possible fixed-sized blockings: 
	\begin{eqnarray*}
		\{\{1,0\}, \{1,0\}, \{1,0\}\} \quad \text{and} \quad \{\{1,1\}, \{1,0\}, \{0,0\}\},
	\end{eqnarray*}
	where the first has a weighted average distance of $(1+1+1)/6=1/2$ while the second, which is optimal, has an average of $(0+1+0)/6=1/6$. There are eight possible threshold blockings, as presented with their aggregated distances in the third column of Table \ref{tab-bestblocking}. The optimal threshold blocking is $\{\{1,1,1\}, \{0,0,0\}\}$ with an average distance of 0. Clearly this is better than the optimal fixed-sized blocking's average of $1/6$.
	
	\subsection{Example 2: Conditional estimator variance} \label{ex-condvar}
	
	Now consider using the conditional variance of the treatment effect estimator as our objective. Unlike the previous example, the choice of randomization method and estimator is no longer immaterial. Suppose treatments are assigned using balanced block randomization, and the effect is estimated using a within-block difference-in-means estimator, both as discussed in \citet{Higgins2015estimators}.
	
	With two treatments, balanced block randomization prescribes that, independently in each block, $\lfloor n_\mathbf{b} / 2 \rfloor$ units are randomly assigned to one of the treatments, picked at random, and $\lceil n_\mathbf{b} / 2 \rceil$ units to the other. If each block contains at least as many units as treatment conditions and there is no attrition, this randomization scheme ensures that the estimator always is defined and unbiased of the true treatment effect. The estimator is defined as:
	\begin{equation}
		\hat{\delta} = \sum_{\mathbf{b}\in\mathbf{B}} \frac{n_\mathbf{b}}{n} \left( \frac{\sum_{i \in \mathbf{b}}T_iy_i}{\sum_{i \in \mathbf{b}}T_i} - \frac{\sum_{i \in \mathbf{b}}(1-T_i)y_i}{\sum_{i \in \mathbf{b}}(1-T_i)}\right), \label{estimator}
	\end{equation}
	where $T_i$ is an indicator of unit $i$'s assigned treatment condition and $y_i$ is its observed response. In other words, the estimator first estimates the effect in each block and then aggregates them to an estimate for the whole sample. 
	
	Now consider using the conditional variance of the estimator as objective: $f(\mathbf{B})=\text{Var}(\hat{\delta}|\mathbf{x},\mathbf{B})$ where $\mathbf{x}$ is the set of all covariates. In Appendix \ref{app-conditional}, I show that, in this setting, the variance is given by:
	\begin{eqnarray*}
		\Var(\hat{\delta}|\mathbf{x},\mathbf{B}) &=& \frac{4}{n} \sum_{\mathbf{b}\in\mathbf{B}} \frac{n_\mathbf{b}}{n} \left(1 + \frac{o_\mathbf{b}}{n_\mathbf{b}^2 - 1}\right) \left(\sigma^2 + s_{x\mathbf{b}}^2\left(\mu_1 - \mu_0\right)^2\right),
		\\[1em]
		s_{x\mathbf{b}}^2 &=& \frac{1}{n_\mathbf{b}-1}\sum_{i \in \mathbf{b}} \left(x_i - \frac{1}{n_\mathbf{b}}\sum_{j \in \mathbf{b}}x_j\right)^2,
	\end{eqnarray*}
	where $o_\mathbf{b}$ is an indicator taking value one if block $\mathbf{b}$ contains an odd number of units, and $\mu_x = \E[y_i(0)|x_i=x]$ is the conditional expectation of the potential outcomes under control treatment. $s_{x\mathbf{b}}^2$ is the (unbiased) sample variance of the covariate in block $\mathbf{b}$ and thereby a measure of within-block covariate homogeneity. The squared difference between the conditional expectations of the potential outcomes, $\left(\mu_1 - \mu_0\right)^2$, acts as a measure of how predictive the covariate is of the outcome.
	
	In this expression $n_\mathbf{b}$, $o_\mathbf{b}$ and $s_{x\mathbf{b}}^2$ are the indirect choice variables, as they are affected by one's choice of blocking, while $n$, $\sigma^2$, $\mu_1$ and $\mu_0$ are (assumed) known parameters. Specifically, we assume we have ex ante knowledge that $\sigma^2=1$ and $\left(\mu_1 - \mu_0\right)^2=2$. Based on the current sample draw we can calculate the variance of each blocking, as presented in the fourth column of Table \ref{tab-bestblocking}.
	
	As seen in the table, the best fixed-sized blocking, $\{\{1,1\}, \{1,0\}, \{0,0\}\}$, produces a conditional variance of $0.889$, while the best threshold blocking, $\{\{1,1,1\}, \{0,0,0\}\}$, produces a lower variance at $0.750$.
\end{proof}

\subsection{Surrogate objective functions} \label{surrogate}

The previous theorems implicitly assume that the objective function is known. In many experiments, the goal of blocking cannot be precisely quantified, or even well-estimated, when the blocks are constructed. The objective function is unknown. The second example above is a common such case: to derive the blockings' variances requires detailed knowledge of the outcome model. With few exceptions, this information is inaccessible. Instead, we must find some other function that we believe captures the relevant features of the true, but inaccessible, objective. When we want to reduce variance, we would typically use some measure of covariate balance. We will investigate this setting in the coming sections, but briefly the idea is that estimator variance depends on how similar the treatment groups are with respect to potential outcomes, and, as units with similar covariate values tend to have similar potential outcomes, striving for covariate balance tends to lower variance.

Borrowing terminology from engineering, I will call any function that takes the place of the true objective function in the optimization for a \emph{surrogate objective function} \cite[see e.g.,][]{Queipo2005}. While one always would prefer to use the true objective, when that is impossible, using some other function, which in some loose sense is associated with the true objective, can provide a good \emph{feasible} alternative. Whenever a surrogate is used, we do not know exactly how blockings map to our objective, and there is no longer a guarantee that threshold blocking yields the best solution.

The performance of a surrogate depends on how well it corresponds to the true objective. If the two functions track each other closely, so that the surrogate's optimum is close to the true optimum, using the surrogate will naturally result in near-optimal blockings. However, whenever the correspondence is not perfect, there can be misleading optimums---sub-optimal blockings which the surrogate wrongly indicates as optimal. When there are such optimums, the method with the best performance in the surrogate does not necessarily lead to the best performance in the true objective.

As discussed in the preceding sections, the difference between threshold and fixed-sized blocking is their search spaces. By having a larger search space, threshold blocking will find a weakly better solution with respect to the surrogate. This might, however, be a misleading optimum. Whenever that is the case, the restricted search space of fixed-sized blocking could shield of the misleading optimums, so that the local optimum in its search space is closer to the true, but unknown, optimum. Generally, when the quality of the surrogate is low, the risk for misleading optimums increases. Thus, the increased search space is likely to be most useful when the surrogate tracks the true objective closely.\footnote{Of course, if the restricted search space was a random subset of the larger space, this would not happen on average. However, when variance is the objective, the search set of fixed-sized blocking differs systematically from that of threshold blocking in aspects relevant to performance.}

As an illustration, consider if we were to use the objective function in the first example as a surrogate for the objective in the second example. The two functions are very similar, albeit not identical. Inspecting Table \ref{tab-bestblocking}, we find a correlation coefficient of 0.9. Being a high quality surrogate, it does not produce misleading minimums---the global minimums of both functions are for the same blocking. Subsequently, the same blocking is produced with the surrogate as with the true objective and, as before, threshold blocking outperforms fixed-sized blocking.

Now consider what happens when we change the predictiveness of the covariate so that $(\mu_1 - \mu_0)^2=0.5$ (from the previous value of 2), effectively making the signal-to-noise ratio, $(\mu_1 - \mu_0)^2/\sigma^2$, lower. As discussed in the coming section, one of the most important factors governing this surrogate's quality is how predictive the covariates are of the outcome. Lowering the signal-to-noise ratio therefore decreases the quality of the surrogate, as indicated by a correlation coefficient of only 0.65. The change does not affect the covariates or their balance, thus the surrogate suggests the same blockings as when $(\mu_1 - \mu_0)^2=2$. The decrease in quality has, however, introduced a misleading optimum. Specifically, the variances of the blockings suggested by the surrogate are now:
\begin{eqnarray*}
	\Var\left(\hat{\delta}\middle|\mathbf{x},\mathbf{B}=\{\{1,1\}, \{1,0\}, \{0,0\}\}\right)& = & 0.722,
	\\
	\Var\left(\hat{\delta}\middle|\mathbf{x},\mathbf{B}=\{\{1,1,1\}, \{0,0,0\}\}\right)& = & 0.750.
\end{eqnarray*}
While with a narrow margin, fixed-sized blocking produces a lower variance. The surrogate's minimum at $\{\{1,1,1\}, \{0,0,0\}\}$ is misleading as the minimum of the true objective is at  $\{\{1,1\}, \{1,0\}, \{0,0\}\}$. Using fixed-sized blocking removes the misleading blocking from the search space, and it can find the true optimum.

\section{Unconditional variance as objective} \label{sec-variance}

The typical blocking scenario is when the scientist employs blocking to reduce variance and uses covariate balance as a surrogate. In this section, I will provide a closer investigation of the determinants of the performance of blocking in that setting. Following most previous studies, I investigate the unconditional estimator variance \citep[see, e.g.,][and the references therein]{Bruhn2009}. While blockings are derived after covariates are observed, which would motivate a focus on the conditional variance, there are two good reasons why the unconditional version is of greatest interest.

A conditional variance is always relative to some sample draw. The performance with one sample is, however, not necessarily representative of other draws: the conditional variance is often sensitive to small differences in the composition of units. In fact, as discussed in Appendix \ref{app-conditional-worse}, one can often construct examples where any blocking method would lead to both higher and lower conditional variance than most other methods, including no blocking. Our conclusions would in that case depend on our choice of sample. The unconditional variance avoids such situations, and allows us to make the most general comparisons. Furthermore, scientists should be interested to commit to an experimental design before collecting their samples as this can greatly improve the credibility of the findings \citep[see, e.g.,][]{Miguel2014}. One must then choose blocking method before observing the covariates, making the unconditional variance the relevant measure.

Unlike previous analytical comparisons of the unconditional estimator variances between blocking methods \citep[e.g.,][]{Abadie2008,Imai2008,Imbens2011}, I do not assume sampling of ready-made blocks. Instead, I consider experiments using ordinary random sampling of units, so that the blocking is a function of the samples' covariate distributions. Assuming block-level sampling is not only at odds with the typical experiment, it also hides some critical aspects. First, different blocking methods require different block structures. For example, fixed-sized blocking requires all blocks to be of a certain size, while threshold blocking allows variable-sized blocks. If we assume sampling of blocks, the same sampling methods cannot be used in both cases as they cannot reproduce the implied structure for both of the methods. Any comparison would thereby be affected by changes both in the blocking and sampling methods.

Second, even if the same block-level sampling method could be used for several blocking methods, the assumption presumes a certain block quality. In reality, quality is a function of the experimental design, i.e., exactly what is studied. For example, it is affected by the choice of surrogate, sample size and, most relevant here, blocking method. Assuming that blocks are sampled would disregard difference in these aspects, unless the assumed sampling method is adjusted accordingly---something that would be equivalent to assuming unit-level sampling in the first place. These problems become particularly troublesome when one assumes sampling of certain number of identical, or near-identical, units with respect to their covariates. This assumption guarantees that homogeneous fixed-sized blocks can be formed and thereby disregards the key disadvantage of fixed-sized block: the strict structure almost always makes such blocks impossible.

While ordinary random sampling brings us closer to the typical experiment and provides some essential insights, it severely complicates the analysis. By assuming block-level sampling, one does not need to be bothered by how the blocks are formed; with unit sampling, we need to derive the exact mapping from observed covariates to blocks to get closed-form expressions. This task is far from trivial. Generally, the only viable route is to restrict the focus to simple covariate distributions, as I do when such expressions are derived in the first part of this section. In the second part, we focus on general properties of this mapping and need not derive the exact blocking for every possible sample draw.

To illustrate how the methods can affect the unconditional variance, I will start my investigation by revisiting a discussion on the performance of paired blocking in the past literature and show how threshold blocking enters into it. I then continue by deriving a decomposition of the unconditional variance for any blocking method using balanced block randomization and the within-block difference-in-means estimator. That is, a decomposition that is valid for all common blocking methods. This shows that the performance depends on primarily three factors: the informational content of the covariates with respect to the outcome (i.e., how predictive they are), to which degree the method can use this information (i.e., the quality of the surrogate and the method's ability to optimize it), and how much variation the method introduces in the number of treated in each block.

\subsection{Threshold blocking can be both best and worst} \label{bestworst}

A common recommendation when designing experiments is that one should always block one's samples. This is, for example, captured by the following quote from \citet[p. 48]{Imai2009}. While speaking of experiments with treatments at the cluster level, their argument is applicable also to individual level treatments:
\begin{quote}
	``[R]andomization by cluster without prior construction of matched pairs, when pairing is feasible, is an exercise in self-destruction.''
\end{quote}
\noindent This recommendation often presumes the use of paired blocking (i.e., matched pairs), and is then motivated by that the blocking cannot result in the higher variance than no blocking and that it, when covariates are informative, leads to the greatest reduction in variance \citep[see, e.g., the discussion in][]{Imbens2011}.\footnote{This statement is somewhat delicate. \citet{Imai2008} shows that paired blocking can produce a higher unconditional variance than no blocking if there is an expected negative correlation in the potential outcomes within pairs. A negative correlation implies a method of forming pairs that is worse than random matching. Apart from bizarre approaches that actively seek to decrease covariate balance, it hard to imagine such ill-performing methods. Even the most naive methods will increase covariate balance on average.}

To see how threshold blocking fits into this discussion, consider a situation similar to the examples above. In an experiment with two treatment conditions, we draw a random sample of $n$ units. We restrict $n$ to be even to facilitate paired blocking. We observe a single binary covariate and use the average Euclidean distance to form the blocks, as above. Treatments are assigned using balanced block randomization, and effects are estimated with the within-block difference-in-means estimator.

Unlike in the previous section, we can no longer consider particular sample draws. As the unconditional variance is the expectation over all possible samples, we must instead focus on the full covariate distribution. For simplicity, assume that $x_i$ is an independent fair coin, so that $\Pr(x_i=1)=0.5$. Furthermore, as the optimal blockings may differ between samples, we must derive the blocking each sample would produce (i.e., we cannot condition on $\mathbf{B}$). That is, we must explicitly include the mapping between the covariate distribution in the sample and the resulting blocking for each method. We will here consider three methods: complete randomization (i.e., no blocking), denoted with $\mathcal{C}$; fixed-sized blocking with a size requirement of two (i.e., paired blocking), denoted with $\mathcal{F}_2$; and threshold blocking also with a size requirement of two, denoted with $\mathcal{T}_2$.

In Appendix \ref{app-unconditional}, I show that, when making the same three assumptions as in Section \ref{thresholdb-adv}, the normalized unconditional variance for the three designs are given by:
\begin{eqnarray*}
	n\Var(\hat{\delta}_\mathcal{C}|\mathcal{C}) &=& 4\sigma^2 + \left(\mu_1 - \mu_0\right)^2,
	\\ [0.6em]
	n\Var(\hat{\delta}_{\mathcal{F}_2}|\mathcal{F}_2) &=& 4\sigma^2 + \frac{2\left(\mu_1 - \mu_0\right)^2}{n},
	\\ [0.6em]
	n\Var(\hat{\delta}_{\mathcal{T}_2}|\mathcal{T}_2) &=& 4\sigma^2 + \frac{8\left(\mu_1 - \mu_0\right)^2}{2^n} + \frac{3\left(2^{n-1} - 2n\right)\sigma^2}{2^nn},
\end{eqnarray*}
where $\mu_x = \E[y_i(0)|x_i=x]$ and $\sigma^2 = \Var[y_i(0)|x_i]$ are the conditional expectation and variance of the potential outcome, as above.

\begin{proposition} \label{prop-canbeworse}
	Threshold blocking can result in an unconditional estimator variance that is higher than when no blocking is done.
\end{proposition}
\begin{proof}
	Consider the difference between threshold blocking and complete randomization in the current setting:
	\begin{eqnarray*}
		n\Var(\hat{\delta}_{\mathcal{T}_2}|\mathcal{T}_2) - n\Var(\hat{\delta}_\mathcal{C}|\mathcal{C}) &=& \frac{3\left(2^{n-1} - 2n\right)\sigma^2}{2^nn} - \frac{(2^n - 8)\left(\mu_1 - \mu_0\right)^2}{2^n}
	\end{eqnarray*}
	When, for example, the covariate is uninformative, so that $\left(\mu_1 - \mu_0\right)^2=0$, this difference is positive for all $n>4$.
\end{proof}

Proposition \ref{prop-canbeworse} shows that if threshold blocking is used when the covariates are uninformative, the unconditional variance is higher than when no blocking is done. This is the case even if covariate balance is improved, and even when there is no negative correlation in the potential outcomes within the blocks. On the contrary, in this case, fixed-sized blocking can never have higher variance than no blocking (as we have $n\geq 2$):
\begin{eqnarray*}
	n\Var(\hat{\delta}_{\mathcal{F}_2}|\mathcal{F}_2) - n\Var(\hat{\delta}_\mathcal{C}|\mathcal{C}) &=& \frac{(2 - n)\left(\mu_1 - \mu_0\right)^2}{n} \leq 0.
\end{eqnarray*}

\begin{proposition} \label{prop-canbebetter}
	Threshold blocking can result in an unconditional estimator variance that is lower than for fixed-sized blocking.
\end{proposition}
\begin{proof}
	Consider the difference between threshold and fixed-sized blocking in the current setting:
	\begin{eqnarray*}
		n\Var(\hat{\delta}_{\mathcal{T}_2}|\mathcal{T}_2) - n\Var(\hat{\delta}_{\mathcal{F}_2}|\mathcal{F}_2) &=& \frac{2^{n-1} - 2n}{2^{n}n}\left(3\sigma^2 - 4\left(\mu_1 - \mu_0\right)^2\right).
	\end{eqnarray*}
	Whenever $3\sigma^2 < 4\left(\mu_1 - \mu_0\right)^2$ and $n>4$, this difference in negative.
\end{proof}

Proposition \ref{prop-canbebetter} shows that paired blocking is not necessarily the best method. There are situations, for example, in this proof when the covariates are quite predictive, where threshold blocking will result in a lower variance.

The propositions show that there is no one-size-fits-all blocking method. In some situations threshold blocking will be superior to a fixed-sized design, and in other situations, the opposite will be true. The decomposition in the next section will make these results understandable and offers some guidance how to choose blocking method.

\subsection{Decomposing the unconditional variance} \label{sec-decomp}

The following decomposition will show that three factors affect the resulting unconditional variance of blocking methods. It extends beyond the three methods considered so far and applies to all blocking methods using the standard within-block difference-in-mean estimator, no matter how the blocks are formed. The one restriction is that it requires that the methods use covariates to form their blocks. Covariates are, however, meant quite widely, including any pre-experimental information. This is therefore true almost by definition; all variables that are available at the time of blocking are covariates under this definition. Blocking on past observations of the outcome or on an estimated prognostic score \citep{Barrios2014} are, for example, both considered feasible.

I will continue to assume random sampling from an infinite population and constant treatment effects, as these greatly simplify the derivation without clouding the intuitions. I will, however, not make parametric assumptions with respect to either the expected potential outcomes or their variances. Specifically, we draw a random sample of size $n$ and observe some set of covariates for each unit ($\mathbf{x}_i$), but we impose no restrictions on the expected outcome ($\mu_{\mathbf{x}_i}$). Focus will be on an arbitrary design, $\mathcal{D}$, and its normalized unconditional variance, $n\Var(\hat{\delta}_\mathcal{D}|\mathcal{D})$. While we do not need to derive the exact mapping, let $\mathcal{D}(\mathbf{x})$ be a function that gives the blocking that the design would produce from some set of covariates, $\mathbf{x} = \{\mathbf{x}_1, \cdots, \mathbf{x}_n\}$.

To start the investigation we use a rather well-known decomposition of the unconditional variance in an experiment \citep[see, e.g.,][]{Abadie2008,Kallus2013}. The law of total variance allows us to differentiate between the uncertainty that stems from sampling and that from treatment assignment:
\begin{eqnarray*}
	n\Var(\hat{\delta}_\mathcal{D}|\mathcal{D}) &=& n\E\left[\Var(\hat{\delta}_\mathcal{D}|\mathbf{x},\mathcal{D})\right] + n\Var\left[\E(\hat{\delta}_\mathcal{D}|\mathbf{x},\mathcal{D})\right].
\end{eqnarray*}
The first term captures that we cannot know the treatment effect in a particular sample with certainty. If the treatment groups were identical with respect to their potential outcomes, we could derive the effect without error---the groups provide a perfect window into each counterfactual world. However, as we cannot observe all potential outcomes, we can never ensure, or even confirm, that the groups are identical in this aspect. We must concede that there will always be small differences between the groups which, while averaging to zero, will led to variance. This variance is captured by a positive $\Var(\hat{\delta}_\mathcal{D}|\mathbf{x},\mathcal{D})$, and its average over sample draws is the first term.

Even if we somehow could calculate the true treatment effect for each sample draw, so that the first term becomes zero, the estimator would still not be constant. While we might know the effect in the sample at hand, we do not know whether that sample is representative of the population. Much in the same way that a non-causal inference, say the average age in some population, cannot be established from a sample without uncertainty, we cannot do the same in an experiment. The second term captures just this. As all considered designs are unbiased, $\E(\hat{\delta}_\mathcal{D}|\mathbf{x},\mathcal{D})$ is equal to the treatment effect in a sample with covariates $\mathbf{x}$, thus this term gives the variance of the treatment effect with respect to sample draws. 

This classical decomposition connects the unconditional variance to the two main parts of the design. The first term is due to unbalanced treatment groups and can therefore be improved with better assignments. The second term is due to unrepresentative samples and can only be lowered by making the treatment effect in the sample more similar to the effect in the population (e.g., using stratified sampling). As blocking does not change the sample, it can only affect the variance by lowering the first term. The current investigation focuses on that term and shows that it can be further decomposed.

\begin{proposition} \label{prop-decomp}
	Given constant treatment effects, the unconditional variance of the block-weighted difference-in-means estimator with any experimental design using blocking can be decomposed into three terms:
	\begin{eqnarray*}
		n\Var(\hat{\delta}_\mathcal{D}|\mathcal{D}) &=& 4 W_1 + 4 W_2 + 2 W_3,
		\\ [2em]
		W_1 &=& \E\left[\Var\left(y_i(0)\middle|\mathbf{x}_i\right)\right],
		\\ [0.7em]
		W_2 &=& \E\left[\sum_{\mathbf{b}\in\mathcal{D}(\mathbf{x})} \frac{n_\mathbf{b}}{n} s_{\mu\mathbf{b}}^2\right],
		\\ [0.7em]
		W_3 &=& \E\left[\sum_{\mathbf{b}\in\mathcal{D}(\mathbf{x})} \frac{n_\mathbf{b}}{n} \Std\left(\frac{1}{T_\mathbf{b}}\middle|n_\mathbf{b}\right) \E\left(s_{y\mathbf{b}}^2\middle|\mathbf{x}\right)\right],
	\end{eqnarray*}
	where $T_\mathbf{b}$ is the number of treated units; $s_{\mu\mathbf{b}}^2$ is the sample variance of the predicted potential outcome; and $s_{y\mathbf{b}}^2$ is the sample variance of the potential outcome, all in block $\mathbf{b}$:
	\begin{eqnarray*}
		T_\mathbf{b} &=& \sum_{i \in \mathbf{b}}T_i
		\\
		s_{\mu\mathbf{b}}^2 &=& \frac{1}{n_\mathbf{b}-1}\sum_{i \in \mathbf{b}}\left(\E\left[y_i(0)|\mathbf{x}_i\right] - \sum_{j \in \mathbf{b}}\frac{\E\left[y_j(0)|\mathbf{x}_j\right]}{n_\mathbf{b}}\right)^2,
		\\
		s_{y\mathbf{b}}^2 &=& \frac{1}{n_\mathbf{b}-1}\sum_{i\in\mathbf{b}}\left(y_i(0) - \sum_{j\in\mathbf{b}}\frac{y_j(0)}{n_\mathbf{b}}\right)^2.
	\end{eqnarray*}
\end{proposition}
Proposition \ref{prop-decomp} is proven in Appendix \ref{app-decomp}. While this decomposition is slightly more complicated than the previous, it too is rather intuitive. Specifically, it shows that the uncertainty stems from three sources: that the covariates does not provide full information about the potential outcomes ($W_1$), that the blocking methods might not construct perfectly homogeneous blocks ($W_2$), and that blocking might introduce variability in the number of the treated units in the blocks ($W_3$).

\subsubsection{The first term: $W_1$}

Intuitively, how well the treatment groups are balanced, and thereby the estimator variance, will depend on the variance of the potential outcomes---the more variation in the outcomes, the higher the risk of unbalanced treatment groups. In the extreme, when there is no variation, the groups are balanced by construction. With the law of total variance we can break the variance of the potential outcomes into two parts: $\Var\left[y_i(0)\right] = \E\left[\Var\left(y_i(0)\middle|\mathbf{x}_i\right)\right] + \Var\left[\E\left(y_i(0)\middle|\mathbf{x}_i\right)\right]$. Now, consider what we know about the potential outcomes.

As the considered blocking methods construct their blocks based on covariates (broadly defined), the only information we have about the potential outcomes are what the covariates provide. Or, formally, before the experiment is conducted, we can have no more information about unit $i$'s outcome than what is given by $\E\left[y_i(0)\middle|\mathbf{x}_i\right]$. If we employ a method, whatever it might be, that fully exploits this information, any variation between units that can be explained by $\E\left[y_i(0)\middle|\mathbf{x}_i\right]$ will go away. This type of explainable variation is captured by the second term of the outcome variance, $\Var\left[\E\left(y_i(0)\middle|\mathbf{x}_i\right)\right]$. In other words, when we fully exploit the covariate information, the remaining contribution of potential outcome variance will be the first part, $\E\left[\Var\left(y_i(0)\middle|\mathbf{x}_i\right)\right]$. Note that this corresponds exactly to $W_1$. This term captures that any blocking method based on covariates cannot lower the variance below what is made possible from the informational content of the covariates.

\subsubsection{The second term: $W_2$}

The first term established a lower bound---no blocking method can have a variance lower than this. This bound exists because we cannot use the potential outcomes directly. However, to reach the bound we must fully use the information provided by the expected potential outcomes, i.e.\ $\mu_{\mathbf{x}} = \E\left[y_i(0)\middle|\mathbf{x}_i=\mathbf{x}\right]$. Blocking methods use the covariates to form blocks and, subsequently, to fully use the information the blocks must contain units that are identical with respect to their expected outcomes. There are two reasons why such blockings are not constructed.

First, $\mu_{\mathbf{x}_i}$ is only the theoretical informational limit, usually we have considerably less information about the outcome model. Naturally, we cannot use information that we do not have. Second, even if we had the information, due to the required block structure, a perfectly homogeneous blocking might not exist. If, for example, a unit is unique in its expected outcome, it must be blocked with units with other values on $\mu_{\mathbf{x}}$. The second term, $W_2$, captures the variance that stems from these two sources. Whenever we lack the information or possibility to construct homogeneous blocks, we must assign units with different values on $\mu_{\mathbf{x}_i}$ to the same block; we thereby introduce variation in the predicted potential outcomes in the block: a positive $s_{\mu\mathbf{b}}^2$. The second term is the weighted expected value of $s_{\mu\mathbf{b}}^2$ across blocks and thus captures how heterogeneous blockings affect the variance.\footnote{We must take the expectation over the sum, as $s_{\mu\mathbf{b}}^2$ could be correlated both with the number of blocks and their size. If we assume that the sample variance is uncorrelated with the block structure (as with fixed-sized blockings), the second term simply becomes $W_2=\E\left[s_{\mu\mathbf{b}}^2\right]$.}

The second term will have a natural connection to covariate balance as the expected outcome is a function of the covariates. However, such connections are hard to quantify without parametric assumptions. There are, nonetheless, two important conclusions that do hold independently of the outcome model. First, by definition, whenever $\mathbf{x}_i=\mathbf{x}_j$, we have $\mu_{\mathbf{x}_i}=\mu_{\mathbf{x}_j}$. That is, if we can create a homogeneous blocking with respect to the covariates, the blocking must be homogeneous also with respect to the expected outcomes and thus the second term is zero. By perfectly balancing the covariates, we get, no matter the outcome model, a blocking that produces the lowest possible variance (disregarding the third term).

Second, if the covariates are irrelevant with respect to the potential outcome, that is $\E\left(y_i(0)\middle|\mathbf{x}_i\right)=\E\left(y_i(0)\right)$, we have $\mu_{\mathbf{x}_i}=\mu_{\mathbf{x}_j}$ for any $\mathbf{x}_i$ and $\mathbf{x}_j$. In this case, the second term will be zero no matter which blocking method we use. When covariates are irrelevant, all blocking methods are equally good at balancing the blocks.

By imposing some structure on the outcome model, we can derive more precisely the connection between expected outcomes and covariates, thereby gain an illustration of the typical behavior. Assume that the conditional expectation function is linear so that $\mu_{\mathbf{x}} = \alpha + \mathbf{x} \boldsymbol\beta$. As shown at the end of Appendix \ref{app-decomp}, the second term then becomes:
\begin{eqnarray*}
	W_2 &=& \boldsymbol\beta^T\E\left[\sum_{\mathbf{b}\in\mathcal{D}(\mathbf{x})} \frac{n_\mathbf{b}}{n} \mathbf{Q}_\mathbf{b}\right] \boldsymbol\beta,
\end{eqnarray*}
where $\mathbf{Q}_\mathbf{b}$ is the sample covariance matrix of the covariates in a block $\mathbf{b}$. The linear model allows us to separate the effect of the covariate balance ($\mathbf{Q}_\mathbf{b}$) from the effect of their predictiveness ($\boldsymbol\beta$). With this outcome model, any type of improvement in covariate balance (i.e., a $\mathbf{Q}_\mathbf{b}$ closer to zero) reduces the estimator variance. Still, when covariates are irrelevant (i.e., $\boldsymbol\beta=\mathbf{0}$), covariate balance does not affect the variance as the expected potential outcomes already are balanced.

The linear example makes clear that knowledge of the outcome model can greatly improve blockings. In this case, we know that the functional form is linear, but we do not necessarily know $\boldsymbol\beta$. Imbalances in covariates that are very predictive of the outcome (i.e., the corresponding coefficient in $\boldsymbol\beta$ is high in absolute terms) are much worse than imbalances in other covariates; imbalances in the former will lead to larger imbalances in the potential outcomes. We would, therefore, like our blocking method to focus on the most predictive covariates. The way to do this is to use a surrogate that puts more weight on the predictive covariates.

\subsubsection{The third term: $W_3$}

Even if we could construct blocks that are perfectly homogeneous, so the second term becomes zero, we have not necessarily reached the variance bound given by the first term.

Balanced block randomization divides each block as evenly as possible between the treatments. When a block has a size that is a multiple of the number of treatment conditions, this is trivial: just assign equally many units to all treatments. Even divisions are, however, not possible with odd-sized blocks. Instead, the blocks are split up to the nearest multiple, and the excess units are randomly assigned treatments. As a result, the number of units in each treatment condition may vary within a block. For example, consider a block with three units in an experiment with active and placebo treatment. It is impossible to partition this block evenly; we cannot split the excess unit between the two treatments. In this example, there must always be two units with active treatment and one with placebo, or \emph{vice versa}.

Consider the estimation of the potential outcomes when there are such imbalances. If we use an estimator that put equal weight on the units' outcomes (e.g., an ordinary difference-in-means estimator or Horvitz-Thompson-type estimator), each treatment condition in odd-sized blocks will not contribute equally to the estimate. In a block with three units and two treatment conditions, one condition will have twice as many units as the other---thus twice the contribution to the estimator. The composition of the treatment groups will, in these cases, not mirror the complete sample. Of course, as such imbalances are equally likely in both directions, the estimator is still unbiased. Allowing for such imbalances, however, partially disregards the information provided by the blocking and will increase variance.

The within-block difference-in-means estimator (as used in the decomposition) solves this problem by first estimating the effect in each block, and, in a second step, it aggregates the block estimates to an overall estimate. This effectively down-weights units in treatment conditions that are overrepresented within the block, and up-weights underrepresented units. This re-weighting ensures that each treatment group is a good representation of the overall sample. It will, however, also reduce the effective sample size---by varying the units' weights we do not fully use the information each unit provides. This is, of course, necessary if we want to balance the treatment groups, but it comes with a cost. This cost is captured in the third term.

The weight of a unit in active treatment is $1/T_\mathbf{b}$ (which has the same distribution as under control treatment due to symmetry in assignment). The factor $\Std\left(1/T_\mathbf{b}\middle|n_\mathbf{b}\right)$ in $W_3$ is therefore a measure of weight variation and captures its effect on the estimator's variance. With fixed-sized blockings, the units in each block will be split equally between the treatment conditions, and $1/T_\mathbf{b}$ will be constant. Whenever threshold blocking is used, however, the units' weights will differ between assignments; there will be variation in $1/T_\mathbf{b}$. Thus, with fixed-sized blocking, the third term will be zero, while it is typically positive with threshold blocking.

The best performing blocking method is the one that best balances the second and the third terms. The first term is common to all methods and thus not much to do about. There is, however, often a trade-off between the other two. To keep the number of treated units fixed---that is, to set the third term to zero---we must ensure that all blocks are multiples of the treatment conditions. While we can reach quite good balance with such a design, at some point the strict structure will constrain us. The only way to get additional improvements is to allow for odd-sized blocks. But when we do that, we introduce variability in the number of treated. The resulting improvement in covariate balance might lead to a decrease in the variance that offsets any increase in the third term, but this is in no way guaranteed. This trade-off demonstrates that a high quality surrogate and predictive covariates are especially useful with threshold blocking. In such cases, we have better knowledge about $s_{\mu\mathbf{b}}^2$ and can use the added flexibility to achieve improvements that are likely to offset the third term.

The three variance expressions in Section \ref{bestworst} provide a good example of the influence of these terms. With the outcome model in that section, the covariates contain no more information than to lower the variance to $4\sigma^2$. This is the first term of all expressions and corresponds to the term $4W_1$ in the decomposition. In the first expression---when no blocking is done---the second term (i.e., $W_2$) is large, reflecting that we can expect quite considerable imbalances without blocking. That method will, however, hold the treatment groups at a constant size, and the third term is zero. Fixed-sized blocking ensures a better balance, which is reflected in that its second term is much lower. As the blocks, by construction, are divisible with the number of treatments, the third term is zero also in this case. Turning to the last expression, we see that threshold blocking generates the lowest second term as it leads to even better balance. The third term is, however, no longer zero as odd-sized blocks are allowed. As shown in Proposition \ref{prop-canbebetter}, when covariates are predictive, the added balance awarded by threshold blocking offsets the variance increase due to the third term.

\section{Simulation results} \label{sec-sim}

Complementing the discussion in the previous section, I will here present a small simulation study investigating the performance of the blocking methods with two outcome models. As we forgo analytical results, we can allow for a slightly more realistic setting: compared to the previous sections the treatment effects are no longer assumed to be constant. With both models, we draw a random sample and observe a single real valued covariate ($x_i$) that is uniformly distributed in the interval from $-5$ to $5$ in the population. In the first model, the potential outcomes depend on this covariate and a standard normal noise term:
\begin{eqnarray*}
	y_i(1) &=& 2x_{i}^2 + \varepsilon_{1i}, 
	\\
	y_i(0) &=& 1.7x_{i}^2 + \varepsilon_{0i},
	\\
	x_{i} &\sim & \mathcal{U}\left(-5, 5\right),
	\\
	\varepsilon_{1i}, \varepsilon_{0i} &\sim & \mathcal{N}\left(0, 1\right).
\end{eqnarray*}
In the second model, the outcome is given only by the noise term:
\begin{eqnarray*}
	y_i(1) &=& \varepsilon_{1i},
	\\
	y_i(0) &=& \varepsilon_{0i},
	\\
	x_{i} &\sim & \mathcal{U}\left(-5, 5\right),
	\\
	\varepsilon_{1i}, \varepsilon_{0i} &\sim & \mathcal{N}\left(0, 1\right).
\end{eqnarray*}
The relevant difference between the two models is that it is only the first where the covariate provides any information about the outcome and, thus, only where blocking can be useful. 

Blocks will be formed based on the Euclidean distances between units within the blocks, i.e., the same surrogate as above. Four performance measures will then be investigated. The first is simply the expected value of the (surrogate) objective function, $\E\left[f(\mathbf{B})\right]$, i.e., the average within-block covariate distance. As this is used to construct the blocks, Theorem \ref{thres:thm:thres-better} applies, and we know that threshold blocking will exhibit the best performance. The other three measures are different variances of the block difference-in-means estimator: the unconditional variance (referred to as PATE), the variance conditional on covariates (CATE) and the variance conditional on potential outcomes (SATE). Using a more detailed conditioning set (e.g., SATE), removes more of the variance that is unaffected by blocking and thus better highlights differences in performance.

The variance conditional on covariates or potential outcomes (i.e., CATE or SATE) will depend on the specific sample we consider. Such conditioning will, as discussed in the previous section, not provide a good indication of the general performances. A method could perform well only with respect to some samples; comparisons between methods are therefore not necessarily fair. To avoid specifying particular samples, I will focus on the expected conditional variances. Due to unbiasedness, this is equal to the mean square error with respect to the corresponding conditional effect:
\begin{eqnarray*}
	\text{PATE:} && \E\left[\left(\hat{\delta}_\mathcal{D} - \E\left(\hat{\delta}_\mathcal{D}\right)\right)^2\right],
	\\
	\text{CATE:} && \E\left[\left(\hat{\delta}_\mathcal{D} - \E\left(\hat{\delta}_\mathcal{D}\middle| \mathbf{x}\right)\right)^2\right],
	\\
	\text{SATE:} && \E\left[\left(\hat{\delta}_\mathcal{D} - \E\left(\hat{\delta}_\mathcal{D}\middle| y_1(0), \cdots, y_n(0), y_1(1), \cdots, y_n(1)\right)\right)^2\right].
\end{eqnarray*}

\begin{table}[p]
	\centering
	\caption{Threshold blocking is best with informative covariates.} \label{tab-corr}
	\begin{tabularx}{0.8\textwidth}{@{\vrule height11pt  width0pt} p{4cm}CCCC}
		Relative performance & Objective & PATE & CATE & SATE \\ \midrule
		\\ [-2.5ex]
		\multicolumn{5}{@{\vrule height11pt  width0pt} l}{\textbf{Panel A:} Sample size $n=12$. } \\
		Complete rand. & 9.14 & 2.125 & 2.149 & 2.159 \\ 
		Fixed-sized bl. & 1.15 & 1.063 & 1.064 & 1.065 \\ [1.2ex]
		\multicolumn{5}{@{\vrule height11pt  width0pt} l}{\textbf{Panel B:} Sample size $n=24$.} \\
		Complete rand. & 20.048 & 3.865 & 4.053 & 4.137 \\ 
		Fixed-sized bl. &  1.255 & 1.155 & 1.169 & 1.176 \\ [1.2ex]
		\multicolumn{5}{@{\vrule height11pt  width0pt} l}{\textbf{Panel C:} Sample size $n=36$.} \\ 
		Complete rand. & 31.074 & 5.282 & 5.815 & 6.083 \\ 
		Fixed-sized bl. &  1.295 & 1.175 & 1.210 & 1.229 \\ [0.5ex]
		\bottomrule
		\\ [-1.9ex]
		\multicolumn{5}{p{0.76\textwidth}}{\footnotesize \textit{Note:} The table presents the performance of complete randomization and fixed-sized blocking relative to threshold blocking when the covariate are correlated with the potential outcomes, using the first data generating process presented in the text. The rows indicate blocking methods, and each cell is the ratio between the measure for the corresponding method and threshold blocking. The columns indicate different measures, where the first is the average value of the objective function, and the three following are the variance measures described in the text. The panels indicate different sample sizes. For example, the top rightmost cell shows that complete randomization produces a variance conditional on potential outcomes that is 2.159 times higher than the variance with threshold blocking. Each model has 1,000,000 simulated experiments based on 100,000 unique sample draws.} \\
	\end{tabularx}
	\bigskip
	\caption{But is less good when they are not.} \label{tab-uncorr}
	\begin{tabularx}{0.8\textwidth}{@{\vrule height11pt  width0pt} p{4cm}CCCC}
		Relative performance & Objective & PATE & CATE & SATE \\ \midrule
		\\ [-2.5ex]
		\multicolumn{5}{@{\vrule height11pt  width0pt} l}{\textbf{Panel A:} Sample size $n=12$. } \\
		Complete rand. & 9.14 & 0.9841 & 0.9841 & 0.9711 \\ 
		Fixed-sized bl. & 1.15 & 0.9850 & 0.9850 & 0.9717 \\ [1.2ex]
		\multicolumn{5}{@{\vrule height11pt  width0pt} l}{\textbf{Panel B:} Sample size $n=24$.} \\
		Complete rand. & 20.048 & 0.9841 & 0.9841 & 0.9703 \\ 
		Fixed-sized bl. &  1.255 & 0.9839 & 0.9839 & 0.9695 \\ [1.2ex]
		\multicolumn{5}{@{\vrule height11pt  width0pt} l}{\textbf{Panel C:} Sample size $n=36$.} \\ 
		Complete rand. & 31.074 & 0.9844 & 0.9844 & 0.9717 \\ 
		Fixed-sized bl. &  1.295 & 0.9842 & 0.9842 & 0.9712 \\ [0.5ex]
		\bottomrule
		\\ [-1.9ex]
		\multicolumn{5}{p{0.76\textwidth}}{\footnotesize \textit{Note:} The table presents the performance of complete randomization and fixed-sized blocking relative to threshold blocking when the covariate are unrelated to the potential outcomes, using the second data generating process presented in the text. See the note of Table \ref{tab-corr} for further details.} \\
	\end{tabularx}
\end{table}

The results, as shown in Table \ref{tab-corr} and \ref{tab-uncorr}, are presented for complete randomization and fixed-sized blocking relative to threshold blocking. For example, a cell with a value of two indicates that the measure for the corresponding method is twice as high as for threshold blocking.

Starting with the first table, we see that threshold blocking produces a lower value for all four measures for every sample size. As the objective function, presented in the first column, is known when the blocks are constructed, there are large improvements when using threshold blocking. Complete randomization has an average value of the objective function that is between 9 and 31 times higher than threshold blocking. Compared to fixed-sized blocking the differences are more modest with between 15\% and 30\% higher values on average. While most of the advantage with blocking occurs already with fixed-sized blocking, these results indicate that non-negligible improvements still can be made.

The three variance measures follow a similar pattern, although the advantage is not as large as for the objective function. This reflects both that covariate imbalances are not the only source of estimator variance, and that the surrogate does not perfectly mirror how covariates are related to the outcome. Still, complete randomization has a variance that is two to six times that of threshold blocking, and fixed-sized blocking is between 6\% and 23\% higher. The more detailed the conditioning set is, the higher the advantage of threshold blocking becomes. Conditioning on more sample information, as with the CATE and SATE, reduces the variance due to sampling but leaves the benefits of blocking intact.

Particularly noteworthy is the relative performance when the sample size increases. For the sizes considered here, threshold blocking performs better as the sample becomes larger.\footnote{As optimal blockings cannot be derived for larger sample sizes in reasonable time (see the concluding remarks), simulation studies are limited to the sizes considered here.} This can be explained by that the search space of threshold blocking grows at a much higher rate than for the other two methods. In other words, threshold blocking has many more opportunities for improvements in large samples. These improvement are also often of the form that a small change in a few blocks cascades through the sample and leads to improvements in many other blocks without changing their size. For illustration, consider a sample if units with covariate values $\{1,3,4,6,7,9,10,12,13, \cdots\}$. Paired blocking must block this as $\{\{1,3\},\{4,6\},\{7,9\},\{10,12\}, \cdots\}$, but adding only two odd-sized blocks, $\{\{1,3,4\},\{6,7\},\{9,10\},\{12,13\}, \cdots\}$ is made possible which has much better balance.

There is, however, an opposing effect when the sample grows bigger. If the support of the covariates is bounded, more units will fill up the covariate space. In other words, as the sample size increases, units' neighbors tend to move closer. While threshold blocking might still confer many improvements, in pure counts, these improvements will not be ``worth'' as much. When the covariate space is densely populated, even sub-optimal blockings tend to lead to rather good balance. At some point, when the covariate space is sufficiently populated, the performances are likely to start to converge. When the dimensionality of this space is high, it will not fill up as fast and convergence occurs later.

Turning to the second table, we see that the improvements in the objective, as presented in the first column, are identical to the first model. This is expected as the covariate distribution and surrogate is unchanged between the models. However, unlike the previous model, the covariate is now completely uninformative of the potential outcomes, and thus improvements in the surrogate do not translate to lower variance. As threshold blocking still introduces the variability in the number of treated units within some blocks, its variance is slightly higher compared to both complete randomization and fixed-sized blocking. This relative increase seems to be constant over different sample sizes.

\section{Concluding remarks} \label{sec-conclude}

When interpreting blocking as a pure optimization problem, the first part of this paper shows that threshold blocking is superior to a fixed-sized design. This interpretation requires that the objective function of true interest is known when the blocks are constructed. There are several situations where this is the case. For example, if blocking is done because of later sub-sample analyses or post-processing steps that require covariate balance, the true objective would be a known function of the covariates. In these cases, threshold blocking is guaranteed to provide the best performance of the two methods.

The second part of the paper shows that this is not necessarily the case when variance reduction is our aim. We cannot calculate the variance that would result from different blockings---the objective function of true interest is not known. In this situation, we must rely on a surrogate, some other function that is related to the objective, to construct the blocks. As we do not know, with certainty, if improvements in the surrogate function translate to improvements in the true objective, maximizing the surrogate might not be beneficial in all situations. How well threshold blocking will perform depends on how closely the surrogate corresponds to the objective.

In the most common case, when one uses covariate balance as a surrogate for unconditional variance, there are several factors that influence performance of blocking methods. First, as blockings are based on covariates, their predictiveness of the outcome will set a bound on how much the variance can be lowered---one cannot remove more uncertainty than what is allowed by the information provided by the covariates. This bound is common to all blocking methods and cannot be affected by this design choice. To lower the bound one must instead collect more pre-experimental information.

Second, the degree to which the blocking methods can use this information will affect their performance. This is governed by the choice of surrogate function. If the surrogate is of high quality (i.e., captures the relevant aspects of the outcome model), blocking are able to take advantage of the information that the covariates provide. Even if the covariates, as a group, are very informative, a badly chosen balance measure would lead to few, if any, improvements in variance. This highlights that one of the most crucial aspects of an experimental design based on blocking is the choice of balance measure. For a given surrogate, different blocking methods will not perform equally well in optimizing it. Specifically, as shown in the first part, threshold blocking will always reach the best blocking with respect to the surrogate.

Last, estimator variance is also affected by variability in the number of treated in the blocks. Ideally, there should be no variation, but if the blocks are not divisible by the number of treatment conditions, this cannot be guaranteed. We can enforce divisibility by constructing fixed-sized blocks. That will, however, restrict our ability to balance the treatment groups: there is a trade-off between this and the second factor. When covariates are quite informative of the outcome, the increased balance made possible by allowing for odd-sized blocks is likely to offset any variance increases due to variability in the number of treated within the blocks.

In principle, it is possible to correct the surrogate for the third factor by penalizing blocks that introduces variability in the number of treated. Doing so would move the surrogate closer to the estimator variance---the true objective---and thus increase its quality. The optimal penalty is, however, relative to the benefit of added covariate balance, which depends on both the outcome model and the sample size (as large samples will have a more densely populated covariate space). While it is doubtful that the optimal penalty can be derived in any specific application, a non-zero penalty is likely to be beneficial with large samples, or when we suspect covariates to be uninformative. From this perspective, fixed-sized blocking can be seen as the special case where the penalty is set to infinity.

A critical factor when choosing blocking method, that has been overlooked so far, is how one finds the optimal blocking. For nearly all samples, methods and surrogates, this is an unwieldy task; there are usually an enormous number of possible blockings. In fact, all the examples in this paper have been chosen so that the optimal blockings either are easy to derive analytically or quickly brute-forced. Currently there exists no fast algorithm that, in a general setting, can derive the optimal blocking either for a threshold or fixed-sized design---that is, an algorithm that terminates in polynomial time.

There are, however, good alternative solutions. For the special case when one seeks a fixed-sized blocking with a required block size of two, there exists an optimal algorithm that runs in polynomial time \citep{Lu2011}. For fixed-sized blockings with block sizes other than two, there exist heuristic algorithms that usually perform well \citep{Moore2012}. For threshold blocking that minimizes the maximum within-block distance, there exists an approximately optimal algorithm that runs in quasilinear time \citep{Higgins2015algorithm}. However, the non-optimality of these algorithms adds another level of complexity to the choice of method. Even in a situation where the optimal blocking of some method is likely to be superior, the same might not hold true for the blockings derived by existing algorithms.

In the end, the choice between fixed-sized and threshold blocking depends both on practical and theoretical considerations. While threshold blocking will under-perform when covariates are uninformative, in most cases where blocking is used to reduce variance, one does so just because the covariates are informative. It is, therefore, likely that threshold blocking would be preferable in many experiments where blocking is believed to be beneficial. However, unless the experiment is very small, this blocking can generally not be found. The real question is which \emph{feasible} method produces the best blocking. Simulation results from \citet{Higgins2015algorithm} indicate that the approximately optimal threshold blocking often performs well. However, when optimal fixed-sized blockings can be found (i.e., in small sample when block sizes of two are desired), it is likely that it will perform best.


\appendix

\section{Deriving the conditional variance} \label{app-conditional}

The following derivation closely follows those of \citet{Higgins2015estimators}, which in turn are based on those in \citet{Cochran1977} and \citet{Lohr1999}. The main difference being that they consider the variance conditional on potential outcomes, while I consider it conditional on covariates.

Let $\hat{\mu}_{\mathbf{b}}(1)$ and $\hat{\mu}_{\mathbf{b}}(0)$ be defined as:
\begin{eqnarray*}
	\hat{\mu}_{\mathbf{b}}(1) & \equiv & \sum_{i \in \mathbf{b}}\frac{T_iy_i}{T_\mathbf{b}} = \sum_{i \in \mathbf{b}}\frac{T_iy_i(1)}{T_\mathbf{b}},
	\\
	\hat{\mu}_{\mathbf{b}}(0) & \equiv & \sum_{i \in \mathbf{b}}\frac{(1-T_i)y_i}{n_\mathbf{b} - T_\mathbf{b}} = \sum_{i \in \mathbf{b}}\frac{(1-T_i)y_i(0)}{n_\mathbf{b} - T_\mathbf{b}},
\end{eqnarray*}
where $T_\mathbf{b}=\sum_{i \in \mathbf{b}}T_i$ and $n_\mathbf{b} - T_\mathbf{b}=\sum_{i \in \mathbf{b}}(1-T_i)$. The estimator (\ref{estimator}) can then be written as:
\begin{equation*}
	\hat{\delta} = \sum_{\mathbf{b}\in\mathbf{B}} \frac{n_\mathbf{b}}{n} \left[\hat{\mu}_{\mathbf{b}}(1) - \hat{\mu}_{\mathbf{b}}(0)\right].
\end{equation*}
When constant treatment effects are assumed we have:
\begin{eqnarray*}
	\hat{\mu}_{\mathbf{b}}(1) & = & \sum_{i \in \mathbf{b}}\frac{T_i(\delta + y_i(0))}{T_\mathbf{b}} = \delta + \hat{\mu}_{\mathbf{b}}'(1),
	\\
	\hat{\mu}_{\mathbf{b}}'(1) & \equiv & \sum_{i \in \mathbf{b}}\frac{T_iy_i(0)}{T_\mathbf{b}},
	\\
	\hat{\delta} &=& \delta + \sum_{\mathbf{b}\in\mathbf{B}} \frac{n_\mathbf{b}}{n} \left[\hat{\mu}_{\mathbf{b}}'(1) - \hat{\mu}_{\mathbf{b}}(0)\right].
\end{eqnarray*}

Treatment is assigned independently across blocks, thus:
\begin{equation*}
	\mathbf{b}_1\neq \mathbf{b}_2 \Rightarrow \Cov[\hat{\mu}_{\mathbf{b}_1}(x), \hat{\mu}_{\mathbf{b}_2}(y)]=0,
\end{equation*}
for any $x$ and $y$. The conditional estimator variance then becomes:
\begin{eqnarray*}
	\Var(\hat{\delta}|\mathbf{x},\mathbf{B}) &=& \Var\left( \delta + \sum_{\mathbf{b}\in\mathbf{B}} \frac{n_\mathbf{b}}{n} \left[\hat{\mu}_{\mathbf{b}}'(1) - \hat{\mu}_{\mathbf{b}}(0)\right]\middle|\mathbf{x},\mathbf{B}\right),
	\\
	&=& \Var\left( \sum_{\mathbf{b}\in\mathbf{B}} \frac{n_\mathbf{b}}{n} \left[\hat{\mu}_{\mathbf{b}}'(1) - \hat{\mu}_{\mathbf{b}}(0)\right]\middle|\mathbf{x},\mathbf{B}\right),
	\\
	&=& \sum_{\mathbf{b}\in\mathbf{B}} \frac{n_\mathbf{b}^2}{n^2} \Var\left[\hat{\mu}_{\mathbf{b}}'(1) - \hat{\mu}_{\mathbf{b}}(0)\middle|\mathbf{x},\mathbf{B}\right],
	\\
	&=& \sum_{\mathbf{b}\in\mathbf{B}} \frac{n_\mathbf{b}^2}{n^2} \left[\Var\left(\hat{\mu}_{\mathbf{b}}'(1)\middle|\mathbf{x},\mathbf{B}\right) + \Var\left(\hat{\mu}_{\mathbf{b}}(0)\middle|\mathbf{x},\mathbf{B}\right)\right.
	\\
	&& \qquad\qquad\left. - 2\Cov\left(\hat{\mu}_{\mathbf{b}}'(1), \hat{\mu}_{\mathbf{b}}(0)\middle|\mathbf{x},\mathbf{B}\right)\right].
\end{eqnarray*}

Under balanced block randomization, all treatment assignments are equally probable; by symmetry we, thereby, have $T_i\sim (1-T_i)$ and $\hat{\mu}_{\mathbf{b}}'(1)\sim\hat{\mu}_{\mathbf{b}}(0)$. This implies that:
\begin{equation*}
	\Var\left[\hat{\mu}_{\mathbf{b}}'(1)|\mathbf{x},\mathbf{B}\right] = \Var\left[\hat{\mu}_{\mathbf{b}}(0)|\mathbf{x},\mathbf{B}\right],
	\quad \text{and} \quad \E\left[\hat{\mu}_{\mathbf{b}}'(1)|\mathbf{x},\mathbf{B}\right] = \E\left[\hat{\mu}_{\mathbf{b}}(0)|\mathbf{x},\mathbf{B}\right],
\end{equation*}
so we have:
\begin{eqnarray*}
	\Var(\hat{\delta}|\mathbf{x},\mathbf{B}) &=& \sum_{\mathbf{b}\in\mathbf{B}} \frac{n_\mathbf{b}^2}{n^2} \left[2\Var\left(\hat{\mu}_{\mathbf{b}}'(1)\middle|\mathbf{x},\mathbf{B}\right) - 2\Cov\left(\hat{\mu}_{\mathbf{b}}'(1), \hat{\mu}_{\mathbf{b}}(0)\middle|\mathbf{x},\mathbf{B}\right)\right],
	\\
	&=& \sum_{\mathbf{b}\in\mathbf{B}} \frac{2n_\mathbf{b}^2}{n^2} \left[\E\left(\left(\hat{\mu}_{\mathbf{b}}'(1)\right)^2\middle|\mathbf{x},\mathbf{B}\right) - \E\left(\hat{\mu}_{\mathbf{b}}'(1)\middle|\mathbf{x},\mathbf{B}\right)^2 \right.
	\\
	&&\qquad\qquad\quad \left. - \E\left(\hat{\mu}_{\mathbf{b}}'(1) \hat{\mu}_{\mathbf{b}}(0)\middle|\mathbf{x},\mathbf{B}\right)\right]
	\\
	&&\qquad\qquad\quad \left. + \E\left(\hat{\mu}_{\mathbf{b}}'(1)\middle|\mathbf{x},\mathbf{B}\right)\E\left(\hat{\mu}_{\mathbf{b}}(0)\middle|\mathbf{x},\mathbf{B}\right)\right],
	\\
	&=& \sum_{\mathbf{b}\in\mathbf{B}} \frac{2n_\mathbf{b}^2}{n^2} \left[\E\left(\left(\hat{\mu}_{\mathbf{b}}'(1)\right)^2\middle|\mathbf{x},\mathbf{B}\right) - \E\left(\hat{\mu}_{\mathbf{b}}'(1)\middle|\mathbf{x},\mathbf{B}\right)^2 \right.
	\\
	&&\qquad\qquad\quad \left. - \E\left(\hat{\mu}_{\mathbf{b}}'(1) \hat{\mu}_{\mathbf{b}}(0)\middle|\mathbf{x},\mathbf{B}\right) + \E\left(\hat{\mu}_{\mathbf{b}}'(1)\middle|\mathbf{x},\mathbf{B}\right)^2\right],
	\\
	&=& \sum_{\mathbf{b}\in\mathbf{B}} \frac{2n_\mathbf{b}^2}{n^2} \left[\E\left(\left(\hat{\mu}_{\mathbf{b}}'(1)\right)^2\middle|\mathbf{x},\mathbf{B}\right) - \E\left(\hat{\mu}_{\mathbf{b}}'(1) \hat{\mu}_{\mathbf{b}}(0)\middle|\mathbf{x},\mathbf{B}\right)\right].
\end{eqnarray*}

Note that treatment assignment is independent of the outcome conditional on covariates and blocking. Furthermore, note that treatment assignment does not depend on the covariates conditional on the blocking, and that the outcome does not depend on the blocking conditional on the covariates. Together with $T_i^2=T_i$ and $T_i(1-T_i)=0$, this yields:
\begin{eqnarray*}
	\E\left(\left(\hat{\mu}_{\mathbf{b}}'(1)\right)^2\middle|\mathbf{x},\mathbf{B}\right) &=& \E\left[\left(\sum_{i \in \mathbf{b}}\frac{T_iy_i(0)}{T_\mathbf{b}}\right)\left(\sum_{i \in \mathbf{b}}\frac{T_iy_i(0)}{T_\mathbf{b}}\right)\middle|\mathbf{x},\mathbf{B}\right],
	\\
	&=& \E\left(\sum_{i \in \mathbf{b}}\sum_{j \in \mathbf{b}}\frac{T_iT_jy_i(0)y_j(0)}{T_\mathbf{b}^2}\middle|\mathbf{x},\mathbf{B}\right),
	\\
	&=& \sum_{i \in \mathbf{b}}\sum_{j \in \mathbf{b}}\E\left(\frac{T_iT_j}{T_\mathbf{b}^2}\middle|\mathbf{B}\right)\E\left(y_i(0)y_j(0)\middle|\mathbf{x}\right),
	\\
	&=& \sum_{i \in \mathbf{b}}\E\left(\frac{T_i}{T_\mathbf{b}^2}\middle|\mathbf{B}\right)\E\left(y_i(0)^2\middle|\mathbf{x}\right)
	\\
	&&\qquad\quad + \sum_{i \in \mathbf{b}}\sum_{j \in \mathbf{b}:j\neq i}\E\left(\frac{T_iT_j}{T_\mathbf{b}^2}\middle|\mathbf{B}\right)\E\left(y_i(0)y_j(0)\middle|\mathbf{x}\right),
	\\
	\E\left(\hat{\mu}_{\mathbf{b}}'(1) \hat{\mu}_{\mathbf{b}}(0)\middle|\mathbf{x},\mathbf{B}\right) &=& \E\left[\left(\sum_{i \in \mathbf{b}}\frac{T_iy_i(0)}{T_\mathbf{b}}\right)\left(\sum_{i \in \mathbf{b}}\frac{(1-T_i)y_i(0)}{n_\mathbf{b} - T_\mathbf{b}}\right)\middle|\mathbf{x},\mathbf{B}\right],
	\\
	&=& \E\left(\sum_{i \in \mathbf{b}}\sum_{j \in \mathbf{b}}\frac{T_i(1-T_j)y_i(0)y_j(0)}{T_\mathbf{b}(n_\mathbf{b} - T_\mathbf{b})}\middle|\mathbf{x},\mathbf{B}\right),
	\\
	&=& \sum_{i \in \mathbf{b}}\sum_{j \in \mathbf{b}}\E\left(\frac{T_i(1-T_j)}{T_\mathbf{b}(n_\mathbf{b} - T_\mathbf{b})}\middle|\mathbf{B}\right)\E\left(y_i(0)y_j(0)\middle|\mathbf{x}\right),
\end{eqnarray*}
\begin{eqnarray*}
	&=& \sum_{i \in \mathbf{b}}\sum_{j \in \mathbf{b}:j\neq i}\E\left(\frac{T_i(1-T_j)}{T_\mathbf{b}(n_\mathbf{b} - T_\mathbf{b})}\middle|\mathbf{B}\right)\E\left(y_i(0)y_j(0)\middle|\mathbf{x}\right).
\end{eqnarray*}
Combining the two expressions we get:
\begin{eqnarray*}
	&& \E\left(\left(\hat{\mu}_{\mathbf{b}}'(1)\right)^2\middle|\mathbf{x},\mathbf{B}\right) - \E\left(\hat{\mu}_{\mathbf{b}}'(1) \hat{\mu}_{\mathbf{b}}(0)\middle|\mathbf{x},\mathbf{B}\right) = 
	\\
	&&\qquad\qquad = \sum_{i \in \mathbf{b}}\E\left(\frac{T_i}{T_\mathbf{b}^2}\middle|\mathbf{B}\right)\E\left(y_i(0)^2\middle|\mathbf{x}\right)
	\\
	&&\qquad\qquad\qquad \, + \sum_{i \in \mathbf{b}}\sum_{j \in \mathbf{b}:j\neq i}\E\left(\frac{T_iT_j}{T_\mathbf{b}^2}\middle|\mathbf{B}\right)\E\left(y_i(0)y_j(0)\middle|\mathbf{x}\right)
	\\
	&&\qquad\qquad\qquad \, - \sum_{i \in \mathbf{b}}\sum_{j \in \mathbf{b}:j\neq i}\E\left(\frac{T_i(1-T_j)}{T_\mathbf{b}(n_\mathbf{b} - T_\mathbf{b})}\middle|\mathbf{B}\right)\E\left(y_i(0)y_j(0)\middle|\mathbf{x}\right),
	\\
	&&\qquad\qquad = \sum_{i \in \mathbf{b}}\E\left(\frac{T_i}{T_\mathbf{b}^2}\middle|\mathbf{B}\right)\E\left(y_i(0)^2\middle|\mathbf{x}\right)
	\\
	&&\qquad\qquad\qquad \, + \sum_{i \in \mathbf{b}}\sum_{j \in \mathbf{b}:j\neq i}\left\{\left[\E\left(\frac{T_iT_j}{T_\mathbf{b}^2}\middle|\mathbf{B}\right)-\E\left(\frac{T_i(1-T_j)}{T_\mathbf{b}(n_\mathbf{b} - T_\mathbf{b})}\middle|\mathbf{B}\right)\right]\right.
	\\
	&&\qquad\qquad\qquad\qquad\qquad\qquad\qquad\left.\times \E\left(y_i(0)y_j(0)\middle|\mathbf{x}\right)\right\}.
\end{eqnarray*}

Consider the expectations containing the treatment indicators. Remember that balanced block randomization, in case with two treatments, mandates that with 50\% probability $T_\mathbf{b}=\lfloor n_\mathbf{b} / 2 \rfloor$ and with 50\% probability it is equal to $\lceil n_\mathbf{b} / 2 \rceil$. By letting $o_\mathbf{b}$ be the remainder when dividing the block size with two ($o_\mathbf{b} \equiv n_\mathbf{b}\bmod 2$), we can write:
\begin{equation*}
	\lfloor n_\mathbf{b} / 2 \rfloor=(n_\mathbf{b} - o_\mathbf{b})/2,
	\quad \text{and} \quad \lceil n_\mathbf{b} / 2 \rceil=(n_\mathbf{b} + o_\mathbf{b})/2.
\end{equation*}
This yields:
\begin{eqnarray*}
	\E\left(\frac{1}{T_\mathbf{b}}\middle|\mathbf{B}\right) &=& \frac{1}{2} \left(\frac{1}{(n_\mathbf{b} - o_\mathbf{b})/2}\right) + \frac{1}{2} \left(\frac{1}{(n_\mathbf{b} + o_\mathbf{b})/2}\right),
	\\
	&=& \frac{1}{n_\mathbf{b} - o_\mathbf{b}} + \frac{1}{n_\mathbf{b} + o_\mathbf{b}},
	\\
	&=& \frac{2n_\mathbf{b}}{n_\mathbf{b}^2 - o_\mathbf{b}}.
\end{eqnarray*}
Note that for a given the number of treated in a block ($T_\mathbf{b}$), the probability for $T_i=1$ is simply the number of treated over the number of units in the block, $T_\mathbf{b}/n_\mathbf{b}$. Together with the law of iterated expectations, this yields: 
\begin{eqnarray*}
	\E\left(\frac{T_i}{T_\mathbf{b}^2}\middle|\mathbf{B}\right) &=& \E\left(\E\left(\frac{T_i}{T_\mathbf{b}^2}\middle|T_\mathbf{b},\mathbf{B}\right)\middle|\mathbf{B}\right),
	\\
	&=& \E\left(\frac{T_\mathbf{b}/n_\mathbf{b}}{T_\mathbf{b}^2}\middle|\mathbf{B}\right),
\end{eqnarray*}
\begin{eqnarray*}
	&=& \frac{1}{n_\mathbf{b}}\E\left(\frac{1}{T_\mathbf{b}}\middle|\mathbf{B}\right),
	\\
	&=& \frac{1}{n_\mathbf{b}}\frac{2n_\mathbf{b}}{n_\mathbf{b}^2 - o_\mathbf{b}} = \frac{2}{n_\mathbf{b}^2 - o_\mathbf{b}}.
\end{eqnarray*}
Similarly, the probability that two units both have $T_i=1$, conditional on the number of treated in a block, is $[T_\mathbf{b}/n_\mathbf{b}]\times [(T_\mathbf{b} - 1)/(n_\mathbf{b}-1)]$. For $i\neq j$, this implies:
\begin{eqnarray*}
	\E\left(\frac{T_iT_j}{T_\mathbf{b}^2}\middle|\mathbf{B}\right) &=& \E\left(\E\left(\frac{T_iT_j}{T_\mathbf{b}^2}\middle|T_\mathbf{b},\mathbf{B}\right)\middle|\mathbf{B}\right),
	\\
	&=& \E\left(\frac{T_\mathbf{b}(T_\mathbf{b} - 1)/(n_\mathbf{b}(n_\mathbf{b}-1))}{T_\mathbf{b}^2}\middle|\mathbf{B}\right),
	\\
	&=& \frac{1}{n_\mathbf{b}(n_\mathbf{b}-1)}\E\left(\frac{T_\mathbf{b} - 1}{T_\mathbf{b}}\middle|\mathbf{B}\right),
	\\
	&=& \frac{1}{n_\mathbf{b}(n_\mathbf{b}-1)}\left[1 - \E\left(\frac{1}{T_\mathbf{b}}\middle|\mathbf{B}\right)\right],
	\\
	&=& \frac{1}{n_\mathbf{b}(n_\mathbf{b}-1)}\left[1 - \frac{2n_\mathbf{b}}{n_\mathbf{b}^2 - o_\mathbf{b}}\right],
	\\
	&=& \frac{1}{n_\mathbf{b}(n_\mathbf{b}-1)} - \frac{2}{(n_\mathbf{b}-1)(n_\mathbf{b}^2 - o_\mathbf{b})}.
\end{eqnarray*}

Finally, the probability that one unit has $T_i=1$ while the other has $T_j=0$, again conditional on the number of treated in a block, is $[T_\mathbf{b}/n_\mathbf{b}]\times [(n_\mathbf{b} - T_\mathbf{b})/(n_\mathbf{b}-1)]$. Then, for $i\neq j$:
\begin{eqnarray*}
	\E\left(\frac{T_i(1-T_j)}{T_\mathbf{b}(n_\mathbf{b} - T_\mathbf{b})}\middle|\mathbf{B}\right) &=& \E\left(\E\left(\frac{T_i(1-T_j)}{T_\mathbf{b}(n_\mathbf{b} - T_\mathbf{b})}\middle|T_\mathbf{b},\mathbf{B}\right)\middle|\mathbf{B}\right),
	\\
	&=& \E\left(\frac{T_\mathbf{b}(n_\mathbf{b} - T_\mathbf{b})/(n_\mathbf{b}(n_\mathbf{b}-1))}{T_\mathbf{b}(n_\mathbf{b} - T_\mathbf{b})}\middle|\mathbf{B}\right),
	\\
	&=& \frac{1}{n_\mathbf{b}(n_\mathbf{b}-1)},
\end{eqnarray*}
\begin{eqnarray*}
	&& \E\left(\frac{T_iT_j}{T_\mathbf{b}^2}\middle|\mathbf{B}\right)-\E\left(\frac{T_i(1-T_j)}{T_\mathbf{b}(n_\mathbf{b} - T_\mathbf{b})}\middle|\mathbf{B}\right) =
	\\
	&&\qquad\qquad = \left(\frac{1}{n_\mathbf{b}(n_\mathbf{b}-1)} - \frac{2}{(n_\mathbf{b}-1)(n_\mathbf{b}^2 - o_\mathbf{b})}\right) - \left(\frac{1}{n_\mathbf{b}(n_\mathbf{b}-1)}\right),
	\\
	&&\qquad\qquad = - \frac{2}{(n_\mathbf{b}-1)(n_\mathbf{b}^2 - o_\mathbf{b})}.
\end{eqnarray*}
Returning to the difference in the expectations we have:
\begin{eqnarray*}
	&& \E\left(\left(\hat{\mu}_{\mathbf{b}}'(1)\right)^2\middle|\mathbf{x},\mathbf{B}\right) - \E\left(\hat{\mu}_{\mathbf{b}}'(1) \hat{\mu}_{\mathbf{b}}(0)\middle|\mathbf{x},\mathbf{B}\right) =
\end{eqnarray*}
\begin{eqnarray*}
	&&\qquad\qquad = \sum_{i \in \mathbf{b}}\frac{2}{n_\mathbf{b}^2 - o_\mathbf{b}}\E\left(y_i(0)^2\middle|\mathbf{x}\right)
	\\
	&&\qquad\quad\qquad\qquad\, + \sum_{i \in \mathbf{b}}\sum_{j \in \mathbf{b}:j\neq i}\left[- \frac{2}{(n_\mathbf{b}-1)(n_\mathbf{b}^2 - o_\mathbf{b})}\right]\E\left(y_i(0)y_j(0)\middle|\mathbf{x}\right),
	\\
	&&\qquad\qquad = \frac{1}{n_\mathbf{b}^2 - o_\mathbf{b}}\left[2\sum_{i \in \mathbf{b}}\E\left(y_i(0)^2\middle|\mathbf{x}\right) + \frac{1}{n_\mathbf{b}-1}\sum_{i \in \mathbf{b}}\sum_{j \in \mathbf{b}:j\neq i}-2\E\left(y_i(0)y_j(0)\middle|\mathbf{x}\right)\right].
\end{eqnarray*}

Note that:
\begin{equation*}
	\E\left(y_i(0)^2\middle|\mathbf{x}\right) = \Var\left(y_i(0)\middle|\mathbf{x}\right) + \E\left(y_i(0)\middle|\mathbf{x}\right)^2,
\end{equation*}
by the definition of variances. Also note that:
\begin{equation*}
	\E\left(y_i(0)y_j(0)\middle|\mathbf{x}\right)= \E\left(y_i(0)\middle|\mathbf{x}\right)\E\left(y_j(0)\middle|\mathbf{x}\right),
\end{equation*}
when $i\neq j$, since random sampling from an infinite population implies that $\Cov\left(y_i(0),y_j(0)\middle|\mathbf{x}\right)=0$. This yields:
\begin{eqnarray*}
	&& \E\left(\left(\hat{\mu}_{\mathbf{b}}'(1)\right)^2\middle|\mathbf{x},\mathbf{B}\right) - \E\left(\hat{\mu}_{\mathbf{b}}'(1) \hat{\mu}_{\mathbf{b}}(0)\middle|\mathbf{x},\mathbf{B}\right) =
	\\
	&&\qquad\qquad = \frac{1}{n_\mathbf{b}^2 - o_\mathbf{b}}\left[2\sum_{i \in \mathbf{b}}\left(\sigma_{x_i}^2 + \mu_{x_i}^2\right) + \frac{1}{n_\mathbf{b}-1}\sum_{i \in \mathbf{b}}\sum_{j \in \mathbf{b}:j\neq i}-2\mu_{x_i}\mu_{x_j}\right],
\end{eqnarray*}
where we define:
\begin{eqnarray*}
	\mu_x&=&\E\left(y_i(0)\middle|x_i=x\right)=\E\left(y_i(0)\middle|\mathbf{x}\right),
	\\
	\sigma_x^2&=&\Var\left(y_i(0)\middle|x_i=x\right)=\Var\left(y_i(0)\middle|\mathbf{x}\right).
\end{eqnarray*}

Consider the second term within the parenthesis:
\begin{eqnarray*}
	&& \frac{1}{n_\mathbf{b}-1}\sum_{i \in \mathbf{b}}\sum_{j \in \mathbf{b}:j\neq i}-2\mu_{x_i}\mu_{x_j} =
	\\
	&&\qquad\qquad = \frac{1}{n_\mathbf{b}-1}\sum_{i \in \mathbf{b}}\sum_{j \in \mathbf{b}:j\neq i}\left[-2\mu_{x_i}\mu_{x_j} + \mu_{x_i}^2 - \mu_{x_i}^2 + \mu_{x_j}^2 - \mu_{x_j}^2 \right],
	\\
	&&\qquad\qquad = \frac{1}{n_\mathbf{b}-1}\sum_{i \in \mathbf{b}}\sum_{j \in \mathbf{b}:j\neq i}\left[\mu_{x_i}^2 - 2\mu_{x_i}\mu_{x_j} + \mu_{x_j}^2  \right]
	\\
	&&\qquad\qquad\qquad\, -\; \frac{1}{n_\mathbf{b}-1}\sum_{i \in \mathbf{b}}\sum_{j \in \mathbf{b}:j\neq i}\left[ \mu_{x_i}^2 + \mu_{x_j}^2 \right],
	\\
	&&\qquad\qquad = \frac{1}{n_\mathbf{b}-1}\sum_{i \in \mathbf{b}}\sum_{j \in \mathbf{b}:j\neq i}\left(\mu_{x_i} - \mu_{x_j}\right)^2 - \frac{2}{n_\mathbf{b}-1}\sum_{i \in \mathbf{b}}\sum_{j \in \mathbf{b}:j\neq i}\mu_{x_i}^2,
	\\
	&&\qquad\qquad = \frac{1}{n_\mathbf{b}-1}\sum_{i \in \mathbf{b}}\sum_{j \in \mathbf{b}:j\neq i}\left(\mu_{x_i} - \mu_{x_j}\right)^2 - \frac{2}{n_\mathbf{b}-1}\sum_{i \in \mathbf{b}}\left(n_\mathbf{b}-1\right)\mu_{x_i}^2,
	\\
	&&\qquad\qquad = \frac{1}{n_\mathbf{b}-1}\sum_{i \in \mathbf{b}}\sum_{j \in \mathbf{b}}\left(\mu_{x_i} - \mu_{x_j}\right)^2 - 2\sum_{i \in \mathbf{b}}\mu_{x_i}^2,
\end{eqnarray*}
where the last equality exploits the fact that $\mu_{x_i} - \mu_{x_i}=0$. Substituting the term in the previous expressions, we get:
\begin{eqnarray*}
	&&\E\left(\left(\hat{\mu}_{\mathbf{b}}'(1)\right)^2\middle|\mathbf{x},\mathbf{B}\right)- \E\left(\hat{\mu}_{\mathbf{b}}'(1) \hat{\mu}_{\mathbf{b}}(0)\middle|\mathbf{x},\mathbf{B}\right) = 
	\\
	&&\quad = \frac{1}{n_\mathbf{b}^2 - o_\mathbf{b}}\left[2\sum_{i \in \mathbf{b}}\left(\sigma_{x_i}^2 + \mu_{x_i}^2\right) + \frac{1}{n_\mathbf{b}-1}\sum_{i \in \mathbf{b}}\sum_{j \in \mathbf{b}}\left(\mu_{x_i} - \mu_{x_j}\right)^2 - 2\sum_{i \in \mathbf{b}}\mu_{x_i}^2\right],
	\\
	&&\quad = \frac{1}{n_\mathbf{b}^2 - o_\mathbf{b}}\left[2\sum_{i \in \mathbf{b}}\sigma_{x_i}^2 + \frac{1}{n_\mathbf{b}-1}\sum_{i \in \mathbf{b}}\sum_{j \in \mathbf{b}}\left(\mu_{x_i} - \mu_{x_j}\right)^2\right].
\end{eqnarray*}
Which in the complete variance expression gives:
\begin{eqnarray}
	\Var(\hat{\delta}|\mathbf{x},\mathbf{B}) &=& \sum_{\mathbf{b}\in\mathbf{B}} \frac{2n_\mathbf{b}^2}{n^2} \left[\E\left(\left(\hat{\mu}_{\mathbf{b}}'(1)\right)^2\middle|\mathbf{x},\mathbf{B}\right) - \E\left(\hat{\mu}_{\mathbf{b}}'(1) \hat{\mu}_{\mathbf{b}}(0)\middle|\mathbf{x},\mathbf{B}\right)\right], \nonumber
	\\
	&=& \sum_{\mathbf{b}\in\mathbf{B}} \frac{2n_\mathbf{b}^2}{n^2} \left[\frac{1}{n_\mathbf{b}^2 - o_\mathbf{b}}\left[2\sum_{i \in \mathbf{b}}\sigma_{x_i}^2 + \frac{1}{n_\mathbf{b}-1}\sum_{i \in \mathbf{b}}\sum_{j \in \mathbf{b}}\left(\mu_{x_i} - \mu_{x_j}\right)^2\right]\right], \nonumber
	\\
	&=& \frac{4}{n}\sum_{\mathbf{b}\in\mathbf{B}} \frac{n_\mathbf{b}}{n} \left(\frac{n_\mathbf{b}^2}{n_\mathbf{b}^2 - o_\mathbf{b}}\right)\left[\sum_{i \in \mathbf{b}}\frac{\sigma_{x_i}^2}{n_\mathbf{b}} + \frac{1}{2n_\mathbf{b}(n_\mathbf{b}-1)}\sum_{i \in \mathbf{b}}\sum_{j \in \mathbf{b}}\left(\mu_{x_i} - \mu_{x_j}\right)^2\right], \nonumber
	\\
	&=& \frac{4}{n} \sum_{\mathbf{b}\in\mathbf{B}} \left\{\frac{n_\mathbf{b}}{n} \left(\frac{o_\mathbf{b}n_\mathbf{b}^2}{n_\mathbf{b}^2 - 1} + \frac{(1-o_\mathbf{b})n_\mathbf{b}^2}{n_\mathbf{b}^2}\right)\right. \times \nonumber
	\\
	&&\qquad\qquad\times\left.\left[\sum_{i \in \mathbf{b}}\frac{\sigma_{x_i}^2}{n_\mathbf{b}} + \frac{1}{2n_\mathbf{b}(n_\mathbf{b}-1)}\sum_{i \in \mathbf{b}}\sum_{j \in \mathbf{b}}\left(\mu_{x_i} - \mu_{x_j}\right)^2\right]\right\}, \nonumber
	\\
	&=& \frac{4}{n} \sum_{\mathbf{b}\in\mathbf{B}} \left\{\frac{n_\mathbf{b}}{n} \left(1 + \frac{o_\mathbf{b}}{n_\mathbf{b}^2 - 1}\right)\right. \times \nonumber
	\\
	&&\qquad\qquad\times\left.\left[\sum_{i \in \mathbf{b}}\frac{\sigma_{x_i}^2}{n_\mathbf{b}} + \frac{1}{2n_\mathbf{b}(n_\mathbf{b}-1)}\sum_{i \in \mathbf{b}}\sum_{j \in \mathbf{b}}\left(\mu_{x_i} - \mu_{x_j}\right)^2\right]\right\}. \label{npara-condvar}
\end{eqnarray}

Remember that we assumed that $\sigma_{x_i}^2=\sigma^2$. Also note that $\mu_{x_i} - \mu_{x_j}=0$ when $x_i= x_j$. Since the support of $x_i$ is $\{0,1\}$ we have $x_i=x_i^2$, yielding:
\begin{eqnarray*}
	&& \Var(\hat{\delta}|\mathbf{x},\mathbf{B}) =
	\\
	&& \quad = \frac{4}{n} \sum_{\mathbf{b}\in\mathbf{B}} \left\{\frac{n_\mathbf{b}}{n} \left(1 + \frac{o_\mathbf{b}}{n_\mathbf{b}^2 - 1}\right)\right. \times
	\\
	&& \qquad\qquad\times \left.\left[\sum_{i \in \mathbf{b}}\frac{\sigma^2}{n_\mathbf{b}} + \frac{1}{2n_\mathbf{b}(n_\mathbf{b}-1)}\sum_{i \in \mathbf{b}}\sum_{j \in \mathbf{b}}2x_i(1-x_j)\left(\mu_1 - \mu_0\right)^2\right]\right\},
\end{eqnarray*}
\begin{eqnarray*}
	&& \quad = \frac{4}{n} \sum_{\mathbf{b}\in\mathbf{B}} \frac{n_\mathbf{b}}{n} \left(1 + \frac{o_\mathbf{b}}{n_\mathbf{b}^2 - 1}\right) \left[\sigma^2 + \frac{\left(\mu_1 - \mu_0\right)^2}{n_\mathbf{b}-1}\sum_{i \in \mathbf{b}}\sum_{j \in \mathbf{b}}\frac{x_i(1-x_j)}{n_\mathbf{b}}\right],
	\\
	&& \quad = \frac{4}{n} \sum_{\mathbf{b}\in\mathbf{B}} \frac{n_\mathbf{b}}{n} \left(1 + \frac{o_\mathbf{b}}{n_\mathbf{b}^2 - 1}\right) \left[\sigma^2 + \frac{\left(\mu_1 - \mu_0\right)^2}{n_\mathbf{b}-1}\sum_{i \in \mathbf{b}} \left(x_i - \frac{1}{n_\mathbf{b}}\sum_{j \in \mathbf{b}}x_j\right)^2\right],
	\\
	&& \quad = \frac{4}{n} \sum_{\mathbf{b}\in\mathbf{B}} \frac{n_\mathbf{b}}{n} \left(1 + \frac{o_\mathbf{b}}{n_\mathbf{b}^2 - 1}\right) \left(\sigma^2 + s_{x\mathbf{b}}^2\left(\mu_1 - \mu_0\right)^2\right),
\end{eqnarray*}
where $s_{x\mathbf{b}}^2$, the sample variance in block $\mathbf{b}$, is defined as:
\begin{eqnarray*}
	s_{x\mathbf{b}}^2 &=& \frac{1}{n_\mathbf{b}-1}\sum_{i \in \mathbf{b}} \left(x_i - \frac{1}{n_\mathbf{b}}\sum_{j \in \mathbf{b}}x_j\right)^2.
\end{eqnarray*}

\section{Deriving the unconditional variance} \label{app-unconditional}

First note that when the treatment effect is constant, as here, for any unbiased experimental design, $\mathcal{D}$, the expected value of the estimator, $\E(\hat{\delta}_\mathcal{D}|\mathbf{x},\mathcal{D})$, is constant at $\delta$. With the law of total variance, for all three considered blocking methods, this implies (see Section \ref{sec-decomp}):
\begin{eqnarray*}
	n\Var(\hat{\delta}_\mathcal{D}|\mathcal{D}) &=& \E\left(n\Var(\hat{\delta}_\mathcal{D}|\mathbf{x},\mathcal{D})\right).
\end{eqnarray*}

We must still consider the distribution of the covariate and how samples map to blockings with the different methods. Consider three functions, $\mathcal{C}(\mathbf{x})$, $\mathcal{F}_2(\mathbf{x})$ and $\mathcal{T}_2(\mathbf{x})$, that provide these mappings. For example, as derived in Section \ref{ex-dist-met}, $\mathcal{F}_2\left(\{1,1,1,0,0,0\}\right)=\{\{1,1\},\{1,0\},\{0,0\}\}$. It turns out that this mapping is quite simple for all three methods in the investigated setting. In particular, when restricting our attention to sample of even sizes (so that fixed-sized blockings exist), they are all completely determined by the sum of all units' covariate values, $\Sigma_x=\sum_{i=1}^nx_i$.

As $x_i$ is a binary indicator, $\Sigma_x$ is a binomial random variable with $n$ ``trials'' each with $p=1/2$ probability of success. Remember that for a binomial random variable we have:
\begin{eqnarray*} 
	\Pr(\Sigma_x = 1) &=& np(1-p)^{n-1} = \frac{n}{2^n},
	\\
	\Pr(\Sigma_x = n - 1) &=& np^{n-1}(1-p) = \frac{n}{2^n}.
\end{eqnarray*}
By a simple recursive argument one can also show that $\Pr(\Sigma_x \bmod 2 = 0)=1/2$:
\begin{eqnarray*} 
	\Pr(\Sigma_x \bmod 2 = 0) &=& \Pr\left(x_1 = 0\right)\Pr\left(\left(\sum\nolimits_{i=2}^{n}x_i\right) \bmod 2 = 0\right)
	\\
	&&\quad +\, \Pr\left(x_1 = 1\right)\Pr\left(\left(\sum\nolimits_{i=2}^{n}x_i\right) \bmod 2 = 1\right),
\end{eqnarray*}
\begin{eqnarray*} 
	&=& \frac{1}{2}\Pr\left(\left(\sum\nolimits_{i=2}^{n}x_i\right) \bmod 2 = 0\right)
	\\
	&&\quad +\, \frac{1}{2}\left(1-\Pr\left(\left(\sum\nolimits_{i=2}^{n}x_i\right) \bmod 2 = 0\right)\right),
	\\
	&=& \frac{1}{2},
\end{eqnarray*}
where the first equality exploits that the ``trials'' are independent, and the second equality that integers must be either even or odd.

\subsection{Complete randomization}

Deriving the blocking under complete randomization ($\mathcal{C}$) is trivial as it always makes a single block of the complete sample, $\mathcal{C}(\mathbf{x})=\{\mathbf{U}\}$. As we have restricted the attention to even sample sizes, we have $o_\mathbf{U}=0$, and the results from Appendix \ref{app-conditional} yields:
\begin{eqnarray*}
	n\Var(\hat{\delta}_\mathcal{C}|\mathcal{C}) &=& \E\left(n\Var(\hat{\delta}_\mathcal{C}|\mathbf{x},\mathcal{C})\right) = \E\left(n\Var(\hat{\delta}_\mathcal{C}|\mathbf{x},\mathbf{B} = \{\mathbf{U}\})\right),
	\\
	&=& \E\left(4 \left(\sigma^2 + s_{x\mathbf{U}}^2\left(\mu_1 - \mu_0\right)^2\right)\right),
	\\
	&=& 4 \left(\sigma^2 + \E\left(s_{x\mathbf{U}}^2\right)\left(\mu_1 - \mu_0\right)^2\right).
\end{eqnarray*}
$\E\left(s_{x\mathbf{U}}^2\right)$ is the expected sample variance in the whole sample. From unbiasedness of the sample variance and the variance of a Bernoulli distribution, we have:
\begin{eqnarray*}
	\E\left(s_{x\mathbf{U}}^2\right) = \Var(x_i) = \frac{1}{4}.
\end{eqnarray*}
By substituting this in the expression for the unconditional variance we get:
\begin{eqnarray*}
	n\Var(\hat{\delta}_\mathcal{C}|\mathcal{C}) &=& 4\sigma^2 + \left(\mu_1 - \mu_0\right)^2.
\end{eqnarray*}

\subsection{Fixed-sized blocking}

When deriving fixed-sized blockings we must normally specify which surrogate function (i.e., covariate balance measure) we use. As we have a single binary covariate, this is not necessary in this case. No matter which function we use to capture the balance, the best we can do is to construct as many pairs as possible with the same covariate values. 

As $n$ is even by assumption, when $\Sigma_x$ is even ($\Sigma_x \bmod 2 = 0$), it is possible to create $n/2$ blocks that are homogeneous. It is impossible to construct any other blocking with better balance; the blocks are perfectly balanced. Thus when this is the case, $\mathcal{F}_2(\mathbf{x})$ will consist of $\Sigma_x/2$ copies of $\{1,1\}$ and $(n-\Sigma_x)/2$ copies of $\{0,0\}$. With perfectly homogeneous blocks there is no within-block covariate variation, thus $s_{x\mathbf{b}}=0$ for all $\mathbf{b}$. Furthermore, as the blocks are fixed at a size of two, we have by construction $o_\mathbf{b}=0$. The formula from Appendix \ref{app-conditional} thereby yields:
\begin{eqnarray*}
	n\Var(\hat{\delta}_{\mathcal{F}_2}|\Sigma_x \bmod 2 = 0,\mathcal{F}_2) &=& 4 \sum_{\mathbf{b}\in\mathcal{F}_2(\mathbf{x})} \frac{n_\mathbf{b}}{n} \sigma^2 = 4\sigma^2.
\end{eqnarray*}

When $\Sigma_x$ is odd ($\Sigma_x \bmod 2 = 1$), perfectly homogeneous blocks of size two can no longer be formed. One block, arbitrary labeled $\mathbf{b}'$, must contain one unit with $x_i=1$ and one with $x_i=0$. For this block we have:
\begin{eqnarray*}
	s_{x\mathbf{b}'}^2 &=& \frac{1}{n_{\mathbf{b}'}-1}\sum_{i \in \mathbf{b}'} \left(x_i - \frac{1}{n_{\mathbf{b}'}}\sum_{j \in \mathbf{b}'}x_j\right)^2 = \left(1 - 0.5\right)^2 + \left(0 - 0.5\right)^2 = 0.5.
\end{eqnarray*}
All other blocks can be constructed to be homogeneous.  $\mathcal{F}_2(\mathbf{x})$ will thus consist of one copy of $\{1,0\}$, $(\Sigma_x-1)/2$ copies of $\{1,1\}$ and $(n-\Sigma_x -1)/2$ copies of $\{0,0\}$. Conditional on an odd $\Sigma_x$, the variance becomes:
\begin{eqnarray*}
	n\Var(\hat{\delta}_{\mathcal{F}_2}|\Sigma_x \bmod 2 = 1,\mathcal{F}_2) &=& 4 \sum_{\mathbf{b}\in\mathcal{F}_2(\mathbf{x})} \frac{n_\mathbf{b}}{n} \left(\sigma^2 + s_{x\mathbf{b}}^2\left(\mu_1 - \mu_0\right)^2\right),
	\\
	&=& 4 \sum_{\mathbf{b}\in\mathcal{F}_2(\mathbf{x})} \frac{n_\mathbf{b}}{n} \sigma^2 + 4 \sum_{\mathbf{b}\in\mathcal{F}_2(\mathbf{x})} \frac{n_\mathbf{b}}{n}s_{x\mathbf{b}}^2\left(\mu_1 - \mu_0\right)^2,
	\\
	&=& 4\sigma^2 + 4\frac{2}{n} s_{x\mathbf{b}'}^2\left(\mu_1 - \mu_0\right)^2,
	\\
	&=& 4\sigma^2 + \frac{4\left(\mu_1 - \mu_0\right)^2}{n}.
\end{eqnarray*}

As $\Var(\hat{\delta}_{\mathcal{F}_2}|\mathbf{x},\mathcal{F}_2)$ is determined by $(\Sigma_x \bmod 2)$, we have:
\begin{eqnarray*}
	n\Var(\hat{\delta}_{\mathcal{F}_2}|\mathcal{F}_2) &=& \E\left(n\Var(\hat{\delta}_{\mathcal{F}_2}|\mathbf{x},\mathcal{F}_2)\right),
	\\
	&=& \E\left(n\Var(\hat{\delta}_{\mathcal{F}_2}|\Sigma_x \bmod 2,\mathcal{F}_2)\right),
	\\
	&=& \Pr(\Sigma_x \bmod 2 = 0)n\Var(\hat{\delta}_{\mathcal{F}_2}|\Sigma_x \bmod 2 = 0,\mathcal{F}_2)
	\\
	&&\, +\, \Pr(\Sigma_x \bmod 2 = 1)n\Var(\hat{\delta}_{\mathcal{F}_2}|\Sigma_x \bmod 2 = 1,\mathcal{F}_2).
\end{eqnarray*}
Remember that $\Pr(\Sigma_x \bmod 2 = 0)=\Pr(\Sigma_x \bmod 2 = 1)=1/2$ which, together with the derived conditional expectations, yields:
\begin{eqnarray*}
	n\Var(\hat{\delta}_{\mathcal{F}_2}|\mathcal{F}_2) &=& \frac{1}{2}\left(4\sigma^2\right) + \frac{1}{2}\left(4\sigma^2 + \frac{4\left(\mu_1 - \mu_0\right)^2}{n}\right),
	\\
	&=& 4\sigma^2 + \frac{2\left(\mu_1 - \mu_0\right)^2}{n}.
\end{eqnarray*}

\subsection{Threshold blocking}

As with the previous methods, optimal threshold blockings can be derived simply from $\Sigma_x$. However, unlike before, the optimal blocking is not unique with respect of the covariates.\footnote{The fixed-sized blocking is not unique with respect to units' identities but is unique with respect to covariates.} For example, if the sample consists of four units, all with covariate value of one, both $\{\{1,1,1,1\}\}$ and $\{\{1,1\},\{1,1\}\}$ are optimal threshold blockings. We will break such ties deterministically. Specifically, whenever there is a tie in covariate balance (as judged by the distance metric function in Section \ref{ex-dist-met}), the blocking with the smallest mean block size will be chosen.

In this case, as with fixed-sized blocking, when $\Sigma_x$ is even, the best threshold blocking is to construct $n/2$ perfectly homogeneous pairs:
\begin{eqnarray*}
	n\Var(\hat{\delta}_{\mathcal{T}_2}|\Sigma_x \bmod 2 = 0,\mathcal{T}_2) = 4\sigma^2.
\end{eqnarray*}

When there is an odd number of units, unlike fixed-sized blocking, threshold blocking is not forced to block two units with different covariate values together. Instead, it can make two blocks, one for each covariate value, to be of size three and thereby retain perfectly homogeneous blocks. In other words, when $\Sigma_x \bmod 2 = 1$ we have that $\mathcal{T}_2(\mathbf{x})$ consists of one copy each of $\{1,1,1\}$ and $\{0,0,0\}$, $(\Sigma_x - 3)/2$ copies of $\{1,1\}$ and $(n - \Sigma_x - 3)/2$ copies of $\{0,0\}$. Implicitly, this assumes that there are enough units to form the blocks $\{1,1,1\}$ and $\{0,0,0\}$. If there, for example, only is a single unit with $x_i=1$, it cannot be blocked with two other units that share the covariate value, as there are no other. The size constraint requires us to have at least two units in each block, and we are left with no other choice but to construct a heterogeneous block.

As the sample is of even size, when $\Sigma_x = 1$ or $n - \Sigma_x = 1$ there is one unit that is alone with its covariate value. Threshold blocking will then form blocks as pairs, of which one has mixed units (i.e., $s_{x\mathbf{b}}^2=0.5$), just as fixed-sized blocking would:
\begin{eqnarray*}
	n\Var(\hat{\delta}_{\mathcal{T}_2}|\Sigma_x \in \{1, n - 1\},\mathcal{T}_2) &=& 4\sigma^2 + \frac{4\left(\mu_1 - \mu_0\right)^2}{n}.
\end{eqnarray*}
When there is several units with both covariate values and $\Sigma_x$ is odd, perfectly homogeneous blocks can be formed by making two blocks with size three, $n_{\mathbf{b}'}=n_{\mathbf{b}''}=3$. For these two we have $o_{\mathbf{b}'}=o_{\mathbf{b}''}=1$ which yields the following conditional variance:
\begin{eqnarray*}
	&& n\Var(\hat{\delta}_{\mathcal{T}_2}|\Sigma_x \bmod 2 = 1,\Sigma_x \not\in \{1, n - 1\},\mathcal{T}_2)\; =
	\\
	&& \qquad\qquad =\; 4 \sum_{\mathbf{b}\in\mathcal{T}_2(\mathbf{x})} \frac{n_\mathbf{b}}{n} \left(1 + \frac{o_\mathbf{b}}{n_\mathbf{b}^2 - 1}\right) \sigma^2,
	\\
	&& \qquad\qquad =\; 4\sigma^2  + \frac{3o_{\mathbf{b}'}\sigma^2}{2n} + \frac{3o_{\mathbf{b}''}\sigma^2}{2n},
	\\
	&& \qquad\qquad =\; 4\sigma^2  + \frac{3\sigma^2}{n}.
\end{eqnarray*}

Similarly to the fixed-sized case, as $\Var(\hat{\delta}_{\mathcal{T}_2}|\mathbf{x},\mathcal{T}_2)$ is determined by $\Sigma_x$, we have:
\begin{eqnarray*}
	n\Var(\hat{\delta}_{\mathcal{T}_2}|\mathcal{T}_2) &=& \E\left(n\Var(\hat{\delta}_{\mathcal{T}_2}|\mathbf{x},\mathcal{T}_2)\right),
	\\
	&=& \E\left(n\Var(\hat{\delta}_{\mathcal{T}_2}|\Sigma_x,\mathcal{T}_2)\right),
\end{eqnarray*}
\begin{eqnarray*}
	&=& \Pr(\Sigma_x \bmod 2 = 0)n\Var(\hat{\delta}_{\mathcal{T}_2}|\Sigma_x \bmod 2 = 0,\mathcal{T}_2)
	\\
	&&\, +\, \Pr(\Sigma_x \in \{1, n - 1\})n\Var(\hat{\delta}_{\mathcal{T}_2}|\Sigma_x \in \{1, n - 1\},\mathcal{T}_2)
	\\
	&&\, +\, \Pr(\Sigma_x \bmod 2 = 1,\Sigma_x \not\in \{1, n - 1\})
	\\
	&&\qquad \times\, n\Var(\hat{\delta}_{\mathcal{T}_2}|\Sigma_x \bmod 2 = 1,\Sigma_x \not\in \{1, n - 1\},\mathcal{T}_2).
\end{eqnarray*}
Remember that $\Sigma_x$ is a binomial random variable, and note that $\Sigma_x \in \{1, n - 1\}$ implies that $\Sigma_x \bmod 2 = 1$:
\begin{eqnarray*}
	\Pr(\Sigma_x \in \{1, n - 1\}) &=& \frac{n}{2^n} + \frac{n}{2^n} = \frac{2n}{2^n},
	\\
	\Pr(\Sigma_x \bmod 2 = 0) &=& \Pr(\Sigma_x \bmod 2 = 1) = \frac{1}{2},
	\\
	\Pr(\Sigma_x \bmod 2 = 1,\Sigma_x \not\in \{1, n - 1\}) &=& \Pr(\Sigma_x \bmod 2 = 1) 
	\\
	&& \qquad -\; \Pr(\Sigma_x \in \{1, n - 1\}),
	\\
	&=& \frac{1}{2} - \frac{2n}{2^n},
\end{eqnarray*}
which yields:
\begin{eqnarray*}
	n\Var(\hat{\delta}_{\mathcal{T}_2}|\mathcal{T}_2) &=& \frac{1}{2}\left(4\sigma^2\right) + \left(\frac{2n}{2^n}\right)\left(4\sigma^2 + \frac{4\left(\mu_1 - \mu_0\right)^2}{n}\right)
	\\
	&&\, +\, \left(\frac{1}{2} - \frac{2n}{2^n}\right)\left(4\sigma^2  + \frac{3\sigma^2}{n}\right)
	\\
	&=& 4\sigma^2 + \frac{8\left(\mu_1 - \mu_0\right)^2}{2^n} + \frac{3\left(2^{n-1} - 2n\right)\sigma^2}{2^nn}
\end{eqnarray*}

\section{The sensitivity of conditional variances} \label{app-conditional-worse}

That the variance of the treatment effect estimator conditional on potential outcomes can be higher using a fixed-sized blocking design than with complete randomization has been discussed, independently, by several authors: it is implied by the results in \citet{Kallus2013}, it is proven in \cite{Imai2008}, it is discussed by \cite{Imbens2011} in his mimeo, and Freedman (\citeyear{Freedman2008}) provides an example in a lecture note. The core idea is quite straightforward. By conditioning on potential outcomes we can basically pick any potential outcomes independently on covariates to prove existence. For any deterministic blocking method, pick potential outcomes so to maximize the dispersion within blocks. This will introduce a negative correlation in the blocks and lead to a variance higher than with no blocking.

To my knowledge, no one has yet discussed how fixed-sized blocking performs relative no blocking when conditioning on covariates, and it is, arguably, a bit trickier to construct examples then. We must still induce a negative correlation in the potential outcomes, so that units tend to be more alike units \emph{not} in their own blocks, but we cannot choose the potential outcomes directly as we only condition on covariates. Fixed-sized blocking will improve overall covariate balance, and units with the same covariate values tend to have the same potential outcomes: at first sight, it seems impossible to induce such correlation. This, however, misses that fixed-sized blocking not necessarily improves covariate balance in all covariates---only in the function used to measure covariate balance. For example, in order to achieve balance in some covariate, blocking might lead to less balance in other covariates. If these happen to be particularly informative, they can induce a negative correlation in the potential outcomes. While this will not happen when averaging over the complete distribution of sample draws, for specific covariate draws, it could.

As an illustration, consider the following experiment. It is identical to the experiment in Section \ref{thresholdb-adv} apart from that there are now two covariates, $x_{i1}$ and $x_{i2}$, of which the first is a binary variable and the other is an integer (e.g., gender and age). Also in this case, we have two treatments, a size requirement of two, use balanced block randomization, the block difference-in-mean estimator and the average Euclidean within-block distance as surrogate. The outcome model, unbeknownst to us, is also the same as in Section \ref{thresholdb-adv}, so that only the first covariate is associated with the outcome:
\begin{eqnarray*}
	\E\left[y_i(0)\middle|x_{i1}, x_{i2}\right] &=& \E\left[y_i(0)\middle|x_{i1}\right],
	\\
	\E\left[y_i(1)\middle|x_{i1}, x_{i2}\right] &=& \E\left[y_i(1)\middle|x_{i1}\right].
\end{eqnarray*}

Again, we have a sample of six units which turn out to have the following covariate values:
\begin{center}
	\begin{tabular}{@{\vrule height11pt width0pt} ccc}
		$i$ & $x_{i1}$ & $x_{i2}$ \\ \midrule
		1   & 1        & 36        \\
		2   & 1        & 38        \\
		3   & 1        & 40        \\
		4   & 0        & 36        \\
		5   & 0        & 38        \\
		6   & 0        & 40        \\
		\bottomrule
	\end{tabular} 
\end{center}
We derive the Euclidean distances between each possible pair of units, as presented in the following distance matrix, where the rows and columns are ordered by the unit index:
\begin{eqnarray*}
	\left[\begin{array}{cccccc}
		0         & 2        & 4         & 1         & \sqrt{5} & \sqrt{17} \\ 
		2         & 0        & 2         & \sqrt{5}  & 1        & \sqrt{5}  \\ 
		4         & 2        & 0         & \sqrt{17} & \sqrt{5} & 1         \\ 
		1         & \sqrt{5} & \sqrt{17} & 0         & 2        & 4         \\ 
		\sqrt{5}  & 1        & \sqrt{5}  & 2         & 0        & 2         \\ 
		\sqrt{17} & \sqrt{5} & 1         & 4         & 2        & 0         \\ 
	\end{array}\right].
\end{eqnarray*}

There are 15 possible fixed-sized blockings in this case. Using the average within-block Euclidean distance---the objective from the example in Section \ref{ex-dist-met}---we can derive the value of the surrogate for each of these blockings, as presented in the following table:
\begin{center}
	\begin{tabular}{@{\vrule height13pt width0pt} lll}
		Blocking & \multicolumn{2}{l}{Distance} \\ \hline
		\\ [-1.5em]
		$\{\{1,2\},\{3,4\},\{5,6\}\}$ & $\left(2 + \sqrt{17} + 2\right)/6$ & $= 1.354$ \\
		$\{\{1,2\},\{3,5\},\{4,6\}\}$ & $\left(2 + \sqrt{5} + 4\right)/6$ & $= 1.373$   \\
		$\{\{1,2\},\{3,6\},\{4,5\}\}$ & $\left(2 + 1 + 2\right)/6$ & $= 0.833$   \\
		$\{\{1,3\},\{2,4\},\{5,6\}\}$ & $\left(4 + \sqrt{5} + 2\right)/6$ & $= 1.373$   \\
		$\{\{1,3\},\{2,5\},\{4,6\}\}$ & $\left(4 + 1 + 4\right)/6$ & $= 1.500$   \\
		$\{\{1,3\},\{2,6\},\{4,5\}\}$ & $\left(4 + \sqrt{5} + 2\right)/6$ & $= 1.373$   \\
		$\{\{1,4\},\{2,3\},\{5,6\}\}$ & $\left(1 + 2 + 2\right)/6$ & $= 0.833$   \\
		$\{\{1,4\},\{2,5\},\{3,6\}\}$ & $\left(1 + 1 + 1\right)/6$ & $= 0.500 \quad \boldsymbol\leftarrow$ \\
		$\{\{1,4\},\{2,6\},\{3,5\}\}$ & $\left(1 + \sqrt{5} + \sqrt{5}\right)/6$ & $= 0.912$   \\
		$\{\{1,5\},\{2,3\},\{4,6\}\}$ & $\left(\sqrt{5} + 2 + 4\right)/6$ & $= 1.373$   \\
		$\{\{1,5\},\{2,4\},\{3,6\}\}$ & $\left(\sqrt{5} + \sqrt{5} + 1\right)/6$ & $= 0.912$   \\
		$\{\{1,5\},\{2,6\},\{3,4\}\}$ & $\left(\sqrt{5} + \sqrt{5} + \sqrt{17}\right)/6$ & $= 1.433$   \\
		$\{\{1,6\},\{2,3\},\{4,5\}\}$ & $\left(\sqrt{17} + 2 + 2\right)/6$ & $= 1.354$   \\
		$\{\{1,6\},\{2,4\},\{3,5\}\}$ & $\left(\sqrt{17} + \sqrt{5} + \sqrt{5}\right)/6$ & $= 1.433$   \\
		$\{\{1,6\},\{2,5\},\{3,4\}\}$ & $\left(\sqrt{17} + 1 + \sqrt{17}\right)/6$ & $= 1.541$   \\ [0.5em] \hline
	\end{tabular}
\end{center}
The blockings are here described by the unit indices rather than their covariate values, as this is less cumbersome with multivariate covariates. The blocking that produces the lowest average distance is $\{\{1,4\},\{2,5\},\{3,6\}\}$, as indicate by the arrow in the table. According to the surrogate, there is no other way to make the covariate more balanced. Note, however, that this blocking maximizes the imbalance in the first covariate---each block contain two units with different values on the binary covariate---the only variable associated with the outcome. Due to the scale of the covariates, imbalances in the second are considered worse than those in the first. All effort is therefore put to balance the second covariate, explaining the resulting blocking.\footnote{Using a distance metric accounting for the scale of the variables, e.g., the Mahalanobis metric, would solve cases like these. Other examples could, however, then still constructed.}

As the outcome model is identical to that in Section \ref{ex-condvar}, the variance formula from that section still applies:
\begin{eqnarray*}
	\Var(\hat{\delta}_{\mathcal{F}_2}|\mathbf{x}=\mathbf{x}',\mathcal{F}_2) &=& \Var(\hat{\delta}|\mathbf{x}=\mathbf{x}',\mathbf{B}=\{\{1,0\},\{1,0\},\{1,0\}\})
	\\
	&=& \frac{2\sigma^2}{3} + \frac{\left(\mu_1 - \mu_0\right)^2}{3},
	\\ [0.5em]
	\Var(\hat{\delta}_{\mathcal{C}}|\mathbf{x}=\mathbf{x}',\mathcal{C}) &=& \Var(\hat{\delta}|\mathbf{x}=\mathbf{x}',\mathbf{B}=\{\{1,1,1,0,0,0\}\})
	\\
	&=& \frac{2\sigma^2}{3} + \frac{\left(\mu_1 - \mu_0\right)^2}{5},
\end{eqnarray*}
where $\mathbf{x}'$ is sample draw described in the table above, and $\mu_x = \E[y_i(0)|x_{i1}=x]$ and $\sigma^2 = \Var[y_i(0)|x_{i1},x_{i2}]$ are defined as above.

It follows that for any $\left(\mu_1 - \mu_0\right)^2>0$ (i.e., when covariates contain some information), we have:
\begin{eqnarray*}
	\Var(\hat{\delta}_{\mathcal{F}_2}|\mathbf{x}=\mathbf{x}',\mathcal{F}_2) - \Var(\hat{\delta}_{\mathcal{C}}|\mathbf{x}=\mathbf{x}',\mathcal{C}) &=& \frac{2\left(\mu_1 - \mu_0\right)^2}{15} > 0.
\end{eqnarray*}
Clearly, fixed-sized blocking can produce a higher variance, conditional on covariates, than no blocking at all.

\section{Decomposing the unconditional variance} \label{app-decomp}

Consider the normalized unconditional variance of an arbitrary design $\mathcal{D}$. With the law of total variance we have:
\begin{eqnarray*}
	n\Var(\hat{\delta}_\mathcal{D}|\mathcal{D}) &=& n\E\left[\Var(\hat{\delta}_\mathcal{D}|\mathbf{x},\mathcal{D})\right] + n\Var\left[\E(\hat{\delta}_\mathcal{D}|\mathbf{x},\mathcal{D})\right], 
	\\
	&=& \E\left[n\Var(\hat{\delta}_\mathcal{D}|\mathbf{x},\mathcal{D})\right] + n\Var\left[\sum_{i=1}^n \frac{\E(y_i(1)-y_i(0)|x_i)}{n}\right],
	\\
	&=& \E\left[n\Var(\hat{\delta}_\mathcal{D}|\mathbf{x},\mathcal{D})\right],
\end{eqnarray*}
where the second equality follows from unbiasedness of $\mathcal{D}$ for all sample draws and the third equality from the constant treatment effect assumption. We can substitute this for expression (\ref{npara-condvar}) derived in Appendix \ref{app-conditional}:
\begin{eqnarray*}
	n\Var(\hat{\delta}_\mathcal{D}|\mathcal{D}) &=& \E\left[n\Var(\hat{\delta}_\mathcal{D}|\mathbf{x},\mathcal{D})\right],
	\\
	&=& \E\left[4 \sum_{\mathbf{b}\in\mathcal{D}(\mathbf{x})} \left\{ \frac{n_\mathbf{b}}{n} \left(1 + \frac{o_\mathbf{b}}{n_\mathbf{b}^2 - 1}\right) \right. \right. \times
	\\
	&& \qquad\qquad \times\left.\left.\left[\sum_{i \in \mathbf{b}}\frac{\sigma_{x_i}^2}{n_\mathbf{b}} + \frac{1}{2n_\mathbf{b}(n_\mathbf{b}-1)}\sum_{i \in \mathbf{b}}\sum_{j \in \mathbf{b}}\left(\mu_{x_i} - \mu_{x_j}\right)^2\right]\right\}\right],
	\\
	&=& 4\E\left[\sum_{i\in\mathbf{U}}\frac{\sigma_{x_i}^2}{n}\right] + 4\E\left[\sum_{\mathbf{b}\in\mathcal{D}(\mathbf{x})} \frac{n_\mathbf{b}}{n} \left[\frac{1}{2n_\mathbf{b}(n_\mathbf{b}-1)}\sum_{i \in \mathbf{b}}\sum_{j \in \mathbf{b}}\left(\mu_{x_i} - \mu_{x_j}\right)^2\right]\right]
	\\
	&&\qquad\qquad +\; 2\E\left[\sum_{\mathbf{b}\in\mathcal{D}(\mathbf{x})} \left\{ \frac{n_\mathbf{b}}{n} \left(\frac{2o_\mathbf{b}}{n_\mathbf{b}^2 - 1}\right) \right.\right. \times
	\\
	&& \qquad\qquad\times\left.\left.\left[\sum_{i \in \mathbf{b}}\frac{\sigma_{x_i}^2}{n_\mathbf{b}} + \frac{1}{2n_\mathbf{b}(n_\mathbf{b}-1)}\sum_{i \in \mathbf{b}}\sum_{j \in \mathbf{b}}\left(\mu_{x_i} - \mu_{x_j}\right)^2\right]\right\} \right],
\end{eqnarray*}
where $\mathcal{D}(\mathbf{x})$ gives the blocks that design $\mathcal{D}$ constructs with sample draw $\mathbf{x}$.

Remember that $T_\mathbf{b}$ is the number of treated in block $\mathbf{b}$ and, as shown in Appendix \ref{app-conditional}:
\begin{eqnarray*}
	\E\left(\frac{1}{T_\mathbf{b}}\middle|n_\mathbf{b}\right) &=& \frac{2n_\mathbf{b}}{n_\mathbf{b}^2 - o_\mathbf{b}}.
\end{eqnarray*}
Consider:
\begin{eqnarray*}
	\E\left(\frac{1}{T_\mathbf{b}^2}\middle|n_\mathbf{b}\right) &=& \frac{1}{2}\left\lfloor \frac{n_\mathbf{b}}{2} \right\rfloor^{-2} + \frac{1}{2}\left\lceil \frac{n_\mathbf{b}}{2} \right\rceil^{-2},
	\\
	&=& \frac{2}{(n_\mathbf{b} - o_\mathbf{b})^2} + \frac{2}{(n_\mathbf{b} + o_\mathbf{b})^2},
	\\
	&=& \frac{4n_\mathbf{b}^2 + 4o_\mathbf{b}^2}{(n_\mathbf{b}^2 - o_\mathbf{b})^2},
	\\ [0.7em]
	\Std\left(\frac{1}{T_\mathbf{b}}\middle|n_\mathbf{b}\right) &=& \sqrt{\E\left(\frac{1}{T_\mathbf{b}^2}\middle|n_\mathbf{b}\right) - \E\left(\frac{1}{T_\mathbf{b}}\middle|n_\mathbf{b}\right)^2},
	\\
	&=& \sqrt{\frac{4n_\mathbf{b}^2 + 4o_\mathbf{b}^2}{(n_\mathbf{b}^2 - o_\mathbf{b})^2} - \frac{4n_\mathbf{b}^2}{(n_\mathbf{b}^2 - o_\mathbf{b})^2}},
	\\
	&=& \frac{2o_\mathbf{b}}{n_\mathbf{b}^2 - 1},
\end{eqnarray*}
where we have exploited that $o_\mathbf{b}$ is binary.

Consider the expected sample variance of the potential outcome in some block $\mathbf{b}$ conditional on the covariates:
\begin{eqnarray*}
	\E\left(s_{y\mathbf{b}}^2\middle|\mathbf{x}\right) &=& \E\left(\sum_{i\in\mathbf{b}}\frac{\left(y_i(0) - \frac{\sum_{j\in\mathbf{b}}y_j(0)}{n_\mathbf{b}}\right)^2}{n_\mathbf{b}-1}\middle|\mathbf{x}\right),
	\\
	&=& \frac{1}{n_\mathbf{b}-1}\sum_{i\in\mathbf{b}}\E\left(y_i(0)^2\middle|\mathbf{x}\right) - \frac{1}{n_\mathbf{b}(n_\mathbf{b}-1)}\sum_{i\in\mathbf{b}}\sum_{j\in\mathbf{b}}\E\left(y_i(0)y_j(0)\middle|\mathbf{x}\right),
	\\
	&=& \frac{1}{n_\mathbf{b}-1}\sum_{i\in\mathbf{b}}\E\left(y_i(0)^2\middle|\mathbf{x}\right) - \frac{1}{n_\mathbf{b}(n_\mathbf{b}-1)}\left[\sum_{i\in\mathbf{b}}\E\left(y_i(0)^2\middle|\mathbf{x}\right) + \sum_{i\in\mathbf{b}}\sum_{j\in\mathbf{b}:j\neq i}\E\left(y_i(0)y_j(0)\middle|\mathbf{x}\right)\right],
	\\
	&=& \sum_{i\in\mathbf{b}}\frac{\sigma_{x_i}^2 + \mu_{x_i}^2}{n_\mathbf{b}} - \frac{1}{n_\mathbf{b}(n_\mathbf{b}-1)}\sum_{i\in\mathbf{b}}\sum_{j\in\mathbf{b}:j\neq i}\mu_{x_i}\mu_{x_j},
	\\
	&=& \sum_{i\in\mathbf{b}}\frac{\sigma_{x_i}^2 + \mu_{x_i}^2}{n_\mathbf{b}} + \frac{1}{2n_\mathbf{b}(n_\mathbf{b}-1)}\sum_{i\in\mathbf{b}}\sum_{j\in\mathbf{b}:j\neq i}\left(- 2\mu_{x_i}\mu_{x_j} + \mu_{x_i}^2 - \mu_{x_i}^2 + \mu_{x_j}^2 - \mu_{x_j}^2\right),
	\\
	&=& \sum_{i\in\mathbf{b}}\frac{\sigma_{x_i}^2 + \mu_{x_i}^2}{n_\mathbf{b}} + \frac{1}{2n_\mathbf{b}(n_\mathbf{b}-1)}\sum_{i\in\mathbf{b}}\sum_{j\in\mathbf{b}}\left(\mu_{x_i} - \mu_{x_j}\right)^2 - \frac{1}{n_\mathbf{b}(n_\mathbf{b}-1)}\sum_{i\in\mathbf{b}}(n_\mathbf{b} - 1)\mu_{x_i}^2,
	\\
	&=& \sum_{i\in\mathbf{b}}\frac{\sigma_{x_i}^2}{n_\mathbf{b}} + \frac{1}{2n_\mathbf{b}(n_\mathbf{b}-1)}\sum_{i\in\mathbf{b}}\sum_{j\in\mathbf{b}}\left(\mu_{x_i} - \mu_{x_j}\right)^2.
\end{eqnarray*}
Further consider the sample variance of the conditional expectation of the potential outcome:
\begin{eqnarray*}
	s_{\mu\mathbf{b}}^2 &=& \frac{1}{n_\mathbf{b}-1}\sum_{i \in \mathbf{b}}\left(\mu_{x_i} - \sum_{j \in \mathbf{b}}\frac{\mu_{x_j}}{n_\mathbf{b}}\right)^2,
	\\
	&=& \frac{1}{n_\mathbf{b}-1}\sum_{i \in \mathbf{b}}\left(\mu_{x_i}^2 - \sum_{j \in \mathbf{b}}\frac{\mu_{x_i}\mu_{x_j}}{n_\mathbf{b}}\right),
	\\
	&=& \frac{1}{n_\mathbf{b}(n_\mathbf{b}-1)}\sum_{i \in \mathbf{b}}\sum_{j \in \mathbf{b}}\left(\mu_{x_i}^2 - \mu_{x_i}\mu_{x_j}\right),
	\\
	&=& \frac{1}{2n_\mathbf{b}(n_\mathbf{b}-1)}\sum_{i \in \mathbf{b}}\sum_{j \in \mathbf{b}}\left(\mu_{x_i}^2 - 2\mu_{x_i}\mu_{x_j} + \mu_{x_j}^2\right),
	\\
	&=& \frac{1}{2n_\mathbf{b}(n_\mathbf{b}-1)}\sum_{i \in \mathbf{b}}\sum_{j \in \mathbf{b}}\left(\mu_{x_i} - \mu_{x_j}\right)^2.
\end{eqnarray*}
Substituting these parts into the variance expression we get:
\begin{eqnarray*}
	n\Var(\hat{\delta}_\mathcal{D}|\mathcal{D}) &=& 4\E\left[\sum_{i\in\mathbf{U}}\frac{\sigma_{x_i}^2}{n}\right] + 4\E\left[\sum_{\mathbf{b}\in\mathcal{D}(\mathbf{x})} \frac{n_\mathbf{b}}{n} s_{\mu\mathbf{b}}^2\right] + 2\E\left[\sum_{\mathbf{b}\in\mathcal{D}(\mathbf{x})} \frac{n_\mathbf{b}}{n} \Std\left(\frac{1}{T_\mathbf{b}}\middle|n_\mathbf{b}\right) \E\left(s_{y\mathbf{b}}^2\middle|\mathbf{x}\right)\right],
	\\
	&=& 4\E\left(\sigma_{x_i}^2\right) + 4\E\left[\sum_{\mathbf{b}\in\mathcal{D}(\mathbf{x})} \frac{n_\mathbf{b}}{n} s_{\mu\mathbf{b}}^2\right] + 2\E\left[\sum_{\mathbf{b}\in\mathcal{D}(\mathbf{x})} \frac{n_\mathbf{b}}{n} \Std\left(\frac{1}{T_\mathbf{b}}\middle|n_\mathbf{b}\right) \E\left(s_{y\mathbf{b}}^2\middle|\mathbf{x}\right)\right].
\end{eqnarray*}

Now assume that the conditional expectation function is linear and consider the second term of the variance:
\begin{eqnarray*}
	\mu_{\mathbf{x}} &=& \E\left(y_i(0)\middle|\mathbf{x}_i=\mathbf{x}\right) = \alpha + \mathbf{x} \boldsymbol\beta,
	\\
	s_{\mu\mathbf{b}}^2 &=& \frac{1}{n_\mathbf{b}-1}\sum_{i \in \mathbf{b}}\left((\alpha + \mathbf{x}_i \boldsymbol\beta) - \sum_{j \in \mathbf{b}}\frac{(\alpha + \mathbf{x}_j \boldsymbol\beta)}{n_\mathbf{b}}\right)^2,
	\\
	&=& \frac{1}{n_\mathbf{b}-1}\sum_{i \in \mathbf{b}} \boldsymbol\beta^T\left(\mathbf{x}_i - \sum_{j \in \mathbf{b}}\frac{\mathbf{x}_j}{n_\mathbf{b}}\right)^T \left(\mathbf{x}_i - \sum_{j \in \mathbf{b}}\frac{\mathbf{x}_j}{n_\mathbf{b}}\right) \boldsymbol\beta,
	\\
	&=& \boldsymbol\beta^T\mathbf{Q}_\mathbf{b} \boldsymbol\beta,
	\\
	\mathbf{Q}_\mathbf{b} &=& \frac{1}{n_\mathbf{b}-1}\sum_{i \in \mathbf{b}}\left(\mathbf{x}_i - \sum_{j \in \mathbf{b}}\frac{\mathbf{x}_j}{n_\mathbf{b}}\right)^T\left(\mathbf{x}_i - \sum_{j \in \mathbf{b}}\frac{\mathbf{x}_j}{n_\mathbf{b}}\right),
	\\
	4\E\left[\sum_{\mathbf{b}\in\mathcal{D}(\mathbf{x})} \frac{n_\mathbf{b}}{n} s_{\mu\mathbf{b}}^2\right] &=& 4\boldsymbol\beta^T\E\left[\sum_{\mathbf{b}\in\mathcal{D}(\mathbf{x})} \frac{n_\mathbf{b}}{n} \mathbf{Q}_\mathbf{b}\right] \boldsymbol\beta.
\end{eqnarray*}
	
\end{document}